%% file: main.tex
\RequirePackage{fix-cm}
\documentclass[smallcondensed]{svjour3}     
\smartqed  
\usepackage{graphicx}
\usepackage{tabularx}
\input{macro}

\journalname{Noname}
\begin{document}

\title{A Formalization of SQL with Nulls\thanks{This research has been supported by the National Cyber Security Centre (NCSC) project: Mechanising the metatheory of SQL with nulls.\\
This work was supported by ERC Consolidator Grant Skye (grant number 682315).}
}



\author{Wilmer Ricciotti \and James Cheney}


\institute{W. Ricciotti (Corresponding author) \and J. Cheney \at
              Laboratory for Foundations of Computer Science,	
              University of Edinburgh,
              10 Crichton St,
              Edinburgh EH8 9AB,
              United Kingdom \\
              \email{research@wilmer-ricciotti.net, jcheney@inf.ed.ac.uk}
}

\date{Received: date / Accepted: date}

\maketitle

\begin{abstract}
SQL is the world's most popular declarative language, forming the
basis of the multi-billion-dollar database industry.  Although SQL
has been standardized, the full standard is based on ambiguous natural
language rather than formal specification.  Commercial SQL
implementations interpret the standard in different ways, so that,
given the same input data, the same query can yield different results
depending on the SQL system it is run on. Even for a particular
system, mechanically checked formalization of all widely-used features of SQL remains an open problem. The lack of a well-understood formal semantics makes it very difficult to validate the soundness of database implementations.
Although formal semantics for fragments of SQL were designed in the past, they usually did not support set and bag operations, lateral joins, nested subqueries, and, crucially, null values. Null values complicate SQL's semantics in profound ways analogous to null pointers or side-effects in other programming languages. Since certain SQL queries are equivalent in the absence of null values, but produce different results when applied to tables containing incomplete data, semantics which ignore null values are able to prove query equivalences that are unsound in realistic databases.
A formal semantics of SQL supporting all the aforementioned features was only proposed recently. In this paper, we report about our mechanization of SQL semantics covering set/bag operations, lateral joins, nested subqueries, and nulls, written in the Coq proof assistant, and describe the validation of key metatheoretic properties.
Additionally, we are able to use the same framework to formalize the semantics of a flat relational calculus (with null values), and show a certified translation of its normal forms into SQL.
\end{abstract}

\clearpage

\input{intro}

\input{sql}

\input{model}

\input{semantics}

\input{rewrite}

\input{tribool}

\input{rc}

\input{related}

\input{concl}

\bibliographystyle{spmpsci}      
\bibliography{main}

\appendix
\input{appendix}


\end{document}

%% file: macro.tex
\usepackage{amsmath,amssymb}
\usepackage{color}
\usepackage{hyperref}
\usepackage{bussproofs}
\usepackage{comment}
\usepackage{listings}
\usepackage{xspace}
\usepackage{stmaryrd}

\newcommand{\HoTTSQL}{$\textsf{HoTT}\textit{SQL}$\xspace}
\newcommand{\SQLNull}{$\textsf{Null}\textit{SQL}$\xspace}
\newcommand{\SQLCoq}{$\textit{SQL}_{\textit{Coq}}$\xspace}
\newcommand{\SQLAlg}{$\textit{SQL}_{\textit{Alg}}$\xspace}

\newcommand{\vect}[1]{\overrightarrow{#1}}
\newcommand{\len}[1]{\vert {#1} \vert}
\newcommand{\bzero}{\mathbf{0}}
\newcommand{\bone}{\mathbf{1}}
\newcommand{\bfB}{\mathbf{B}}
\newcommand{\bfk}{\mathbf{k}}

\newcommand{\bN}{\mathbb{N}}
\newcommand{\app}{\mathop{++}}
\newcommand{\eval}{\mathop{\Downarrow}}

\newcommand{\orelse}{\mathop{~\vert~}}
\newcommand{\kwif}{\mathtt{if}}
\newcommand{\kwthen}{\mathtt{then}}
\newcommand{\kwelse}{\mathtt{else}}
\newcommand{\kwcard}{\mathtt{card}}
\newcommand{\kwlet}{\mathtt{let}}
\newcommand{\kwlin}{\mathtt{in}}
\newcommand{\kwsel}{{\color{blue}\mathtt{SELECT}}}
\newcommand{\kwdist}{{\color{blue}\mathtt{DISTINCT}}}
\newcommand{\kwfrom}{{\color{blue}\mathtt{FROM}}}
\newcommand{\kwlat}{{\color{blue}\mathtt{LATERAL}}}
\newcommand{\kwwhere}{{\color{blue}\mathtt{WHERE}}}
\newcommand{\kwas}{{\color{blue}\mathtt{AS}}}
\newcommand{\kwunion}{{\color{blue}\mathtt{UNION}}}
\newcommand{\kwinters}{{\color{blue}\mathtt{INTERSECT}}}
\newcommand{\kwexcept}{{\color{blue}\mathtt{EXCEPT}}}
\newcommand{\kwsome}{\mathit{some}}
\newcommand{\kwtb}{\mathit{table}}
\newcommand{\kwquery}{\mathit{query}}
\newcommand{\kwall}{{\color{blue}\mathtt{ALL}}}
\newcommand{\kwin}{{\color{blue}\mathtt{IN}}}
\newcommand{\kwis}{{\color{blue}\mathtt{IS}}}
\newcommand{\kwnot}{{\color{blue}\mathtt{NOT}}}
\newcommand{\kwnull}{{\color{blue}\mathtt{NULL}}}
\newcommand{\kwtrue}{{\color{blue}\mathtt{TRUE}}}
\newcommand{\kwfalse}{{\color{blue}\mathtt{FALSE}}}
\newcommand{\kwex}{{\color{blue}\mathtt{EXISTS}}}
\newcommand{\kwand}{{\color{blue}\mathtt{AND}}}
\newcommand{\kwor}{{\color{blue}\mathtt{OR}}}
\newcommand{\kwnodup}{\mathit{nodup}}
\newcommand{\kwtt}{\mathbf{tt}}
\newcommand{\kwff}{\mathbf{ff}}
\newcommand{\kwuu}{\mathbf{uu}}
\newcommand{\kwisnull}{\mathit{isnull}}
\newcommand{\kwistrue}{\mathit{istrue}}
\newcommand{\dom}[1]{\mathit{dom({#1})}}
\newcommand{\flatten}[1]{\mathit{flatten({#1})}}
\newcommand{\snil}{\langle\rangle}
\newcommand{\fnil}{[]}
\newcommand{\twovl}{\mathsf{2VL}}
\newcommand{\threevl}{\mathsf{3VL}}
\newcommand{\subst}[2]{\left\{{#2}\middle/{#1}\right\}}

\newcommand{\sem}[1]{\left\llbracket {#1} \right\rrbracket}
\newcommand{\gsem}[1]{\sem{#1}^\bfB}
\newcommand{\bsem}[1]{\sem{#1}^\twovl}
\newcommand{\tsem}[1]{\sem{#1}^\threevl}

\newcommand{\sfwa}[3]{\kwsel~[\kwdist]~{#1}~\kwfrom~{#2}~\kwwhere~{#3}}
\newcommand{\sfwb}[3]{
  \begin{array}{l}
  \kwsel~[\kwdist]~{#1} 
  \\
  \kwfrom~{#2}~\kwwhere~{#3}
\end{array}
}
\newcommand{\sfwc}[3]{
  \begin{array}{l}
  \kwsel~{#1} 
  \\
  \kwfrom~{#2}~\kwwhere~{#3}
\end{array}
}

\newcommand{\ssfwa}[2]{\kwsel~[\kwdist]~\ast~\kwfrom~{#1}~\kwwhere~{#2}}
\newcommand{\ssfwb}[2]{
  \begin{array}{l}
  \kwsel~[\kwdist]~\ast
  \\
  \kwfrom~{#1}~\kwwhere~{#2}
  \end{array}
}

\newcommand{\jq}[4][D]{{#2} \vdash_{#1} {#3} \Rightarrow {#4}}
\newcommand{\jc}[3][D]{{#2} \vdash_{#1} {#3}}
\newcommand{\jt}{\jc}
\newcommand{\jtl}{\jc}
\newcommand{\jT}{\jq}
\newcommand{\jTl}{\jq}
\newcommand{\jbT}{\jq}
\newcommand{\jiq}{\jc}
\newcommand{\jrcb}{\jc}
\newcommand{\jrcc}[3][D]{\jq[#1]{#2}{#3}{\kwcond}}
\newcommand{\jrct}[4][D]{\jq[#1]{#2}{#3}{\kwtuple~{#4}}}
\newcommand{\jrcs}[4][D]{\jq[#1]{#2}{#3}{\kwcoll~{#4}}}
\newcommand{\jrcd}[4][D]{\jq[#1]{#2}{#3}{\kwdisj~{#4}}}
\newcommand{\jrcg}[4][D]{\jq[#1]{#2}{#3}{\kwgen~{#4}}}

\newcommand{\rVS}{\noalign{\vskip 2ex plus.4ex minus.4ex}}

\definecolor{dred}{RGB}{127,0,0}
\definecolor{dgreen}{RGB}{0,127,0}
\definecolor{wrbk}{rgb}{0.8,0.9,1}
\definecolor{jcbk}{rgb}{0.8,1,0.8}
\definecolor{wrtxt}{rgb}{0,0.5,0.8}
\definecolor{jctxt}{rgb}{0,0.6,0}

\newcommand{\NRC}[0]{\ensuremath{\mathcal{NRC}}\xspace}

\newcommand{\NRClsb}[0]{\ensuremath{\NRC_\lambda(\mathit{Set,Bag})}\xspace}

\newcommand{\tuple}[1]{\langle{#1}\rangle}
\def\comprehension{\bigcup\setlit}

\def\kwset{\mathsf{set}}
\def\kwbag{\mathsf{bag}}
\newcommand{\setlit}[1]{\{{#1}\}}
\def\distinct{\delta}
\def\promote{\iota}

\def\plempty{\mathbf{empty}}

\def\kwtuple{\mathbf{tuple}}
\def\kwcond{\mathbf{cond}}
\def\kwcoll{\mathbf{coll}}
\def\kwdisj{\mathbf{disj}}
\def\kwgen{\mathbf{gen}}

\newcommand{\xlate}[1]{\left\Vert {#1} \right\Vert}

\newcommand{\wrnote}[1]{\colorbox{wrbk}{\parbox{0.9\textwidth}{\textcolor{wrtxt}{WR:} #1}}}

\lstdefinelanguage{coq}
{
frame=single,
frameround=ftft,
basicstyle=\scriptsize\ttfamily,
commentstyle=\color{dred},
morecomment=[n]{(*}{*)},
classoffset=0,
morekeywords={Definition, Lemma, Theorem, Parameter, Hypothesis, Variable, Inductive, Fixpoint},
keywordstyle=\bfseries\color{red},
classoffset=1,
morekeywords={fun, match, with, end, for, forall, exists, Type, Prop},
keywordstyle=\color{dgreen},
sensitive=false,
morestring=[b]",
literate={->}{{$\to$\ }}{1} {=>}{{$\Rightarrow$\ }}{1} {|}{{$\vert$}}{2},
mathescape
}

\lstnewenvironment{coq}{\lstset{language=coq}}{}

\lstdefinelanguage{icoq}
{
frame=single,
frameround=ftft,
basicstyle=\small\ttfamily,
commentstyle=\color{dred},
morecomment=[n]{(*}{*)},
classoffset=0,
morekeywords={Definition, Lemma, Theorem, Parameter, Hypothesis, Variable, Inductive, Fixpoint},
keywordstyle=\bfseries\color{red},
classoffset=1,
morekeywords={fun, match, with, end, for, forall, exists, Type, Prop},
keywordstyle=\color{dgreen},
sensitive=false,
morestring=[b]",
literate={->}{{$\to$}}{1} {=>}{{$\Rightarrow$}}{1} {|}{{$\vert$}}{2},
mathescape
}
\lstMakeShortInline[language=icoq]"

\lstdefinelanguage{sql}
{
frame=single,
frameround=ftft,
basicstyle=\scriptsize\ttfamily,
classoffset=0,
morekeywords={SELECT, DISTINCT, FROM, LATERAL, WHERE, EXISTS, UNION, INTERSECT,
  EXCEPT, ALL, TRUE, FALSE, IS, NOT, NULL, IN, AND, OR, NOT, AS, COUNT},
keywordstyle=\bfseries\color{blue},
sensitive=false,
literate={*}{{$\ast$}}{2},
mathescape
}

\lstnewenvironment{sql}{\lstset{language=sql}}{}

\lstdefinelanguage{isql}
{
frame=single,
frameround=ftft,
basicstyle=\small\ttfamily,
classoffset=0,
morekeywords={SELECT, DISTINCT, FROM, LATERAL, WHERE, EXISTS, UNION, INTERSECT,
  EXCEPT, ALL, TRUE, FALSE, IS, NOT, NULL, IN, AND, OR, NOT, AS, COUNT},
keywordstyle=\bfseries\color{blue},
sensitive=false,
literate={*}{{$\ast$}}{2},
mathescape
}

\lstMakeShortInline[language=isql]|

%% file: intro.tex
\section{Introduction}
SQL is the standard query language used by relational databases, which
are the basis of a multi-billion dollar industry.  SQL's semantics is
notoriously subtle: the standard (ISO/IEC 9075:2016) uses natural
language that implementations interpret in different ways.

Relational databases are the world's most successful example of
declarative programming.  Commercial databases optimize queries by
applying rewriting rules to convert a request into an
equivalent one that can be executed more efficiently, using the
database's knowledge of data organization, statistics, and
indexes. However, the lack of a well-understood formal semantics for
SQL makes it very difficult to validate the soundness of candidate
rewriting rules, and even widely used database systems have been known
to return incorrect results due to bugs in query transformations (such
as the ``COUNT bug'')~\cite{Kim82,Ganski87}.  As a result, many database systems
conservatively use a limited set of very well-understood rewrite rules.

An accurate understanding of the semantics of SQL is also required to validate
techniques used to integrate SQL queries in a host programming language. One such
technique, which has been particularly influential in recent years, is \emph{language-integrated query}:
it is based on a domain specific sublanguage
of the host programming language, whose expressions can be made to correspond, after some
manipulation, to SQL queries. In order for the validity of this correspondence to be verified, we need a
formal semantics of SQL.

One of SQL's key features is incomplete information, i.e. \emph{null
  values}.  Null values are special tokens that indicate a
``missing'' or ``unknown'' value.  Unlike the ``none'' values in ``option'' or ``maybe''
types in functional languages such as ML, Haskell, or Scala, null
values are permitted as values of any field by default unless explicitly ruled
out as part of a table's schema (type declaration).  Moreover,
standard arithmetic and other primitive operations are extended to
support null values, and predicates are extended to
three-valued interpretations, to allow for the possibility that
a relationship cannot be determined to be either true or false due to
null values.  
As a result, the impact of nulls on the semantics of SQL
is similar to that of \emph{effects} such as null pointers,
exceptions, or side-effecting references in other programming
languages: almost any query can have surprising behavior in the presence of nulls.

SQL's idiosyncratic treatment of nulls is a common source of bugs in
database applications and query optimizers, especially in combination
with SQL's \emph{multiset} (or \emph{bag}) semantics. 
For example, consider the following three queries:
\begin{sql}
SELECT * FROM R WHERE 1 = 1
SELECT * FROM R WHERE A = A
SELECT * FROM R WHERE A = B OR A <> B
\end{sql}
over a relation $R$ with fields $A,B$. In conventional two-valued
logic, all three queries are equivalent because the |WHERE|-clauses are
tautologies. However, in the presence of nulls, all three
queries have different behavior: the first simply returns $R$, while
the second returns all elements of $R$ whose $A$-field is nonnull, and
the third returns all elements of $R$ such that both $A$ and $B$
values are nonnull. In the second query, if a record's $A$ value is
null, then the truth value of $A = A$ is \emph{maybe}, and such records are
not included in the resulting set. Likewise, if one of $A$ or $B$ (or
both!) is null, then $A = B 
\vee A \neq B$ has truth value
\emph{maybe}.

This problem, unfortunately, pervades most SQL features, even ones
that do not explicitly refer to equality tests.  For example, in the
absence of null values, Guagliardo and Libkin observe that all three
of the following queries have equivalent behavior (\cite{guagliardo17}):
\begin{sql}
SELECT R.A FROM R WHERE R.A NOT IN (SELECT S.A FROM S)
SELECT R.A FROM R WHERE NOT EXISTS (SELECT * FROM S WHERE S.A = R.A)
SELECT R.A FROM R EXCEPT SELECT S.A FROM S
\end{sql}
but all three have \emph{different} behavior when presented with the input table $R = \{1,null\}$ and $S = \{null\}$.  The first results in $\emptyset$, the second in $\{1,null\}$, and the third in $\{1\}$.

SQL's rather
counterintuitive semantics involving NULLs and three-valued logic
leads query optimizers to be conservative in order to avoid subtle bugs. Database implementations tend to restrict attention to a
small set of rules that have been both carefully proved correct (on
paper) and whose correctness has been validated over time. This means
that to get the best performance, a programmer often needs to know
what kinds of optimizations the query optimizer will perform and how
to reformulate queries to ensure that helpful optimizations take place. Of course, this merely
passes the buck: now the programmer must reason about the correctness
or equivalence of the more-efficient query, and as we have seen this
is easy to get wrong in the presence of nulls. As a result, database
applications are either less efficient or less reliable than they
should be.

Formal verification and certification of query transformations offers
a potential solution to this problem. We envision a (not too distant)
future in which query optimizers are \emph{certified}: that is, in
addition to mapping a given query to a hopefully more efficient one,
the optimizer provides a checkable proof that the two queries are
equivalent. Note that (as with certifying compilers~\cite{compcert})
this does not require proving the correctness or even termination of
the optimizer itself. Furthermore, we might consider several
optimizers, each specializing in different kinds of queries.

Before we get too excited about this vision, we should recognize that
there are many obstacles to realizing it.  For example, before we can
talk about proving the correctness of query transformations, let alone
mechanically checked proofs, we need to have a suitable semantics of
queries.  Formal semantics for SQL has been investigated
intermittently, including mechanized formalizations and proofs;
however, most such efforts have focused on simplified core
languages with no support for nulls~\cite{malecha10popl,benzaken14esop,Chu17},
meaning that they can and do prove equivalences that are false in real
databases, which invariably do support nulls (a recent exception to this is \SQLCoq~\cite{Benzaken19}, which we will discuss later).
Part of the reason for neglecting nulls and three-valued logic is that the theory of relational databases
and queries has been developed largely in terms of the
\emph{relational algebra} which does not support such concepts. Recent work by Guagliardo and Libkin~\cite{guagliardo17} provides the first (on-paper) formal
semantics of SQL with nulls (we will call this \SQLNull). \SQLNull is
the first formal treatment of SQL's nulls and three-valued semantics, and it
has been validated empirically using random testing to compare with
the behaviour of real database engines, but mechanized formalizations of the semantics of SQL with nulls have only appeared recently.

\subsubsection*{Contributions}\label{sec:hottsql}
This paper is a report about our formalization of SQL with null values, three-valued logic, and lateral joins: our development can be publicly accessed
at its GitHub repository
(\url{https://github.com/wricciot/nullSQL}). The most complete
formalization of SQL to date is \SQLCoq~\cite{Benzaken19}, which
was developed concurrently with our work: it formalizes a variant of
\SQLNull with grouping and aggregates and a corresponding bag-valued
relational algebra, proving the equivalence between the two. Our work
does not deal with grouping and aggregation; however, it does provide
a more accurate formalization of well-formedness constraints for SQL expressions. 
The well-formedness judgment defined in \SQLCoq accepts queries using free attribute 
names (not bound to an input table), which are rejected by concrete implementations; 
in the formalization, such queries are assigned a dummy semantics in the form of default values.

Another relevant formalization is \HoTTSQL by Chu et al.~\cite{Chu17}, which does not allow incomplete information in tables;
as it turns out, formalizing SQL with nulls requires us to deal with issues that are not immediately evident in \HoTTSQL, and thus provides us with an opportunity to consider alternatives to some of their design choices. 

We summarize here the key features of our formalization compared to the existing work.

\paragraph{Representation of tables.}
%
The \HoTTSQL paper describes two concrete alternatives for the
representation of tables: the list model and the $K$-relation
model~\cite{Green07}. They argue that lists are difficult to reason on because of the
requirement that they be equal up to permutation of elements, and that$K$-relations require the invariant of finite-supportedness to be
wired through each proof. They then go on to extend the $K$-relation
model to $K$ allowing infinite cardinalities (through HoTT types) and
claim this is a substantial improvement; they also use univalent types $\bzero$ and $\bone$ to represent truth values. However, they do not prove an
adequacy property relating this representation to a conventional one.
Despite the ease of reasoning with the \HoTTSQL approach, it is unclear how to adapt it to three-valued logic. 

As for \SQLCoq, \cite{Benzaken19} does not discuss the representation of tables in great detail; however, their formalization uses a bag datatype provided in a Coq library.

In this paper, we show instead that the difficulty of reasoning on lists up to permutations, which partly motivated the recourse to HoTT, is a typical proof-engineering issue, stemming from a lack of separation between the properties that the model is expected to satisfy, and its implementation as data (which is typical of type theory). Our key contribution is, therefore, the definition of $K$-relations as an abstract data type whose inhabitants can only be created, examined, and composed by means of structure-preserving operations, and its concrete implementation as normalized lists.

\paragraph{Reasoning on relations.}
This is a related point. Reasoning on an ADT cannot use the case analysis and induction principles that are normally the bread and butter of Coq users; for this reason, our ADT will expose some abstract well-behavedness properties that can be used as an alternative to concrete reasoning. Additionally, we will assume heterogeneous (\emph{``John Major''}) equality to help with the use of dependent types, and functional extensionality to reason up to rewriting under binders (such as the $\Sigma$ operator of $K$-relations expressing projections -- and more complex maps in our formalization).

\paragraph{The formalized fragment of SQL.}
Aside from nulls, there are several differences between the fragments of SQL used by the three formalizations. To list a few:
\begin{itemize}
\item \HoTTSQL does not employ names at any level, therefore attributes must be referenced in a de~Bruijn-like style, by position in a tuple rather than by name; \SQLCoq uses names for attributes, but not for tables, and relies on the implicit assumption that attributes be renamed so that no aliasing can happen in a cross product; in our formalization, names are used to reference attributes, and de~Bruijn indices to reference tables; our semantics is nameless.
\item Since \HoTTSQL does not have names, it does not allow attributes to be projected
  just by referencing them in a select clause (as we do), but it provides
  additional language expressions to express projections as a
  (forgetful) reshuffling of an input sequence of attributes.
\item \SQLCoq, on the other hand, by assuming that no attribute clash can occur, does not address the attribute shadowing problem mentioned by \cite{guagliardo17}.
\item Both \HoTTSQL and \SQLCoq do consider \emph{grouping} and
  \emph{aggregation} features, which are not covered by \cite{guagliardo17}, nor by our formalization;
\item Unlike both \HoTTSQL and \SQLCoq, we formalize SQL queries with |LATERAL| input, introduced in the SQL:1999 standard and supported by recent versions of DBMSs such as Oracle, PostgreSQL, and MySQL. 
  When a subquery appearing in the |FROM| clause is preceded by |LATERAL|, that subquery is allowed to reference attributes introduced by the preceding |FROM| items: this means that while normally the |FROM| items 
  of a |SELECT| query are evaluated independently, a |LATERAL| subquery needs to be evaluated once for every tuple in the preceding |FROM| items, making its semantics substantially more complicated.
\end{itemize}

\paragraph{Boolean semantics vs. three-valued semantics.}
As we mentioned above, in 
\linebreak
\HoTTSQL the evaluation of the |WHERE| clauses of queries yields necessarily a Boolean value. However, in standard SQL, conditional expressions can evaluate to an uncertain truth value, due to the presence of incomplete information in the data base. The lack of an obvious denotation of the uncertain truth value as a HoTT type makes it challenging to extend that work to nulls even in principle.
Our formalization, like Benzaken and Contejean's, provides a semantics for \SQLNull based on three-valued logic; additionally, we provide a Boolean semantics as well: we can thus formally derive
Guagliardo and Libkin's proof that, even in the presence of nulls,
three-valued logic does not increase the expressive power of SQL, and even extend it to queries with |LATERAL| input.
Whether such a property holds in the presence of grouping and
aggregation does not appear to have been investigated.

\paragraph{Relational calculus vs. SQL.}
The language-integrated query feature of programming languages such as Kleisli~\cite{Won00}, Links~\cite{CLWY06}, and Microsoft's C\# and F\# allows a user to express database queries in a typed domain-specific sublanguage which blends in nicely with the rest of the program. Core calculi such as the nested relational calculus~\cite{BNTW95} (\NRC) have been used to provided a theoretical basis to study language-integrated query: in particular, Wong's conservativity theorem~(\cite{wong:conservativity}) implies that every \NRC query mapping flat tables to flat tables can be \emph{normalized} to a flat relational calculus query, not using nested collections as intermediate data. 
Such flat queries correspond closely to SQL queries, and it is straightforward to give an algorithm to translate the former into the latter. Furthermore, in~\cite{Ricciotti19dbpl} and~\cite{Ricciotti21esop}, we extended \NRC to allow queries mixing set and bag collections, and we noted that in this language, under additional conditions, it is still possible to normalize flat queries to a form that directly corresponds to SQL, as long as |LATERAL| inputs are allowed.
 
However, the correspondence established by these works is rather informal: the correctness of translations from \NRC to SQL has not been proved formally, at least to our knowledge. In Section~\ref{sec:rc}, we fill this gap in the literature: we formally define flat relational calculus normal forms using sets and bags and their semantics, show a translation mapping them to SQL, and prove that the translation preserves the semantics of the original query.

\subsection{Structure of the paper}
We start in Section~\ref{sec:syntax} by describing our formalization of the syntax of \SQLNull, discussing our implementation choices and differences with the official SQL syntax; Section~\ref{sec:relations} is devoted to our semantic model of relations, particularly its implementation as an abstract data type; in Section~\ref{sec:semantics}, we describe how SQL queries are evaluated to semantic relations, using both Boolean and three-valued logic;  Section~\ref{sec:translation} formalizes Guagliardo and Libkin's proof that the two versions of the semantics have the same expressive power; finally Section~\ref{sec:rc} gives a semantics of normalized flat relational calculus terms and gives an algorithm to translate them to SQL queries, proving its correctness.

\section{Overview of the formalization}
The formalization we describe is partitioned in several modules and functors. In some cases, these serve as little more than namespaces, or are used mostly for the purpose of presentational separation. For example, the various parts of this development are defined in terms of an underlying collection of named tables, namely \emph{the data base} $D$; rather than cluttering all the definitions with references to $D$ and its properties, we package their signature in a module type "DB" and assume that a concrete implementation is given.

The syntax of \SQLNull, including rules defining well-formedness of queries and other expressions, is defined in a module of type "SQL".

%% file: sql.tex
\section{Syntax}\label{sec:syntax}
We formalize a fragment of SQL consisting of select-from-where queries
(including ``select-star'') with correlated subqueries connected with
$\kwex$ and $\kwin$ and operations of union, intersection and
difference. Both set and bag (i.e. multiset) semantics are supported, through the use of the keywords $\kwdist$ and $\kwall$.
We assume a simple data model consisting of constants $\bfk$ along with the unspecified $\kwnull$ value. We make no assumption over the semantics of constants, which may thus stand for numeric values, strings, or any other kind of data; however, for the purpose of formalization it is useful to assume that the constants be linearly ordered, for example by means of the lexicographic order on their binary representation. Relations are bags of $n$-tuples of values, where $n$ is the arity of the relation.
Our syntax is close to the standard syntax of SQL, but we make a few simplifying changes:
\begin{itemize}
\item The tables in the |FROM| clause of |SELECT-FROM-WHERE| queries are referenced by a $0$-based de~Bruijn index rather than by name; however, attributes are still referenced by name.
\item Attribute (re)naming using |AS|, both in |SELECT| and |FROM|, is mandatory.
\item The |WHERE| clause is mandatory (|WHERE TRUE| must be used when no condition is given).
\item An explicit syntax ($\kwtb~x$ or $\kwquery~Q$) is provided to differentiate between tables stored by name in the database and tables resulting from a query.
\end{itemize}
Hence, if $R$ is a relation with column names $A$, $B$, $C$, the SQL query 
|SELECT R.A FROM R| 
must be expressed as |SELECT| $0.$|A AS A FROM| $\kwtb$ |R AS (A,B,C) WHERE TRUE|.

For compactness, we will write $\kwas$ as a colon ``:''. The full syntax follows:
\[
\begin{array}{lcl@{\hspace{1.2cm}}lcl@{\hspace{1.2cm}}lcl@{\hspace{1.2cm}}lcl}
x & \in & \mathbb{X}
&
\alpha & ::= & n.x
&
\sigma & ::= & \vect{x}
&
\Gamma & ::= & \vect{\sigma}
\\
t & ::= & \multicolumn{4}{l}{\alpha \orelse \mathbf{k} \orelse \kwnull}
&
v & ::= & \multicolumn{4}{l}{\mathbf{k} \orelse \kwnull}
\\
Q & ::= & \multicolumn{10}{l}{\sfwa{\vect{t : x}}{G}{c}}
  \\
  & \orelse & \multicolumn{10}{l}{\ssfwa{G}{c}}
  \\
  & \orelse & \multicolumn{10}{l}{Q_1~\kwunion~[\kwall]~Q_2 
  \orelse Q_1~\kwinters~[\kwall]~Q_2 
  \orelse Q_1~\kwexcept~[\kwall]~Q_2}
\\
G & ::= & \multicolumn{4}{l}{
  \vect{T_1 : \sigma_1}~\kwlat \ldots \kwlat~\vect{T_n : \sigma_n}
}
\\
c & ::= & \multicolumn{10}{l}{\kwtrue \orelse \kwfalse \orelse c~\kwis~\kwtrue \orelse t~\kwis~[\kwnot]~\kwnull
  \orelse \vect{t}~[\kwnot]~\kwin~Q
  \orelse P^n(\vect{t_n})}
  \\
  & \orelse & 
  \multicolumn{10}{l}{\kwex~Q \orelse c_1~\kwand~c_2 \orelse c_1~\kwor~c_2 \orelse \kwnot~c}
\\
T & ::= & \multicolumn{10}{l}{\kwtb~x \orelse \kwquery~Q}
\end{array}
\]
The |SELECT| clause of a query takes a list of terms, which include null or constant values, and references to attributes one of the tables in the form $n.x$, where $n$ is the index referring to an input relation in the |FROM| clause, and $x$ is an attribute name. The input of the query is expressed by the |FROM| clause, which references a \emph{generator} $G$ consisting of a sequence of \emph{frames} separated by the |LATERAL| keyword; each frame is a sequence $\vect{T:\sigma}$ of input tables paired with a schema (allowing attribute renaming); an input table can be defined using variables introduced in a previous frame, but not in the same frame; concretely, in a query:
\begin{sql}
SELECT z.id
FROM T1 x,
     (SELECT * FROM T2 x' WHERE x.name = x'.name) y,
     LATERAL (SELECT * FROM T3 x' WHERE x.name = x'.name) z
\end{sql}
the expression introducing the variable |y| is ill-formed, because it uses the variable |x|, which is introduced in the same frame; however, the very similar expression associated to |z| is well-formed, because it is part of a different frame introduced by |LATERAL|. In our Coq formalization, we will model frames as lists, and sequences of frames as lists of lists.

Conditions for the |WHERE| clause of queries include Booleans and Boolean operators (|TRUE|, |FALSE|, |AND|, |OR|, |NOT|), comparison of conditions with |TRUE|, comparison of terms with |NULL|, membership tests for tuples ($\vect{t}~[\kwnot]~\kwin~Q$), non-emptiness of the result of subqueries (|EXISTS Q|), and custom predicates $P^n(\vect{t_n})$ (where $P^n$ is an $n$-ary Boolean predicate, and $\vect{t_n}$ an $n$-tuple of terms). 

The abstract syntax we have presented in Section~\ref{sec:syntax} is made concrete in Coq by means of inductive types.
\begin{coq}
  Inductive pretm : Type :=
  | tmconst : BaseConst -> pretm
  | tmnull  : pretm
  | tmvar   : FullVar -> pretm   
  
  Inductive prequery : Type :=
  | select  : bool -> list (pretm * Name) -> list (list (pretb * Scm)) -> 
                precond -> prequery
  | selstar : bool -> list (list (pretb * Scm)) -> precond -> prequery
  | qunion  : bool -> prequery -> prequery -> prequery
  | qinters : bool -> prequery -> prequery -> prequery
  | qexcept : bool -> prequery -> prequery -> prequery

  with precond : Type :=
  | cndtrue   : precond
  | cndfalse  : precond
  | cndnull   : bool -> pretm -> precond
  | cndistrue : precond -> precond
  | cndpred   : forall n, (forall l : list BaseConst, 
                    length l = n -> bool) -> 
                  list pretm -> precond
  | cndmemb   : bool -> list pretm -> prequery -> precond
  | cndex     : prequery -> precond
  | cndand    : precond -> precond -> precond
  | cndor     : precond -> precond -> precond
  | cndnot    : precond -> precond

  with pretb: Type :=
  | tbbase  : Name -> pretb
  | tbquery : prequery -> pretb.
\end{coq}
Query constructors "select" and "selstar" take a Boolean argument which, when it is true, plays the role of a |DISTINCT| selection query; similarly, the Boolean argument to constructors "qunion", "qinters", and "qexcept" plays the role of the |ALL| modifier allowing for union, intersection, and difference according to bag semantics. Conditions using base predicates are expressed by the constructor "cndpred": notice that we do not formally specify the set of base predicates defined by SQL, but allow any $n$-ary function from constant values (of type "BaseConst") to Booleans expressible in Coq to be embedded in an SQL query: such functions can easily represent SQL predicates including equality, inequality, numerical ``greater than'' relations, |LIKE| on strings, and many more.

We use well-formedness judgments (Fig.~\ref{fig:SQLjudg}) to filter out meaningless expressions, in particular those containing table references that cannot be resolved because they point to a table that is not in the |FROM| clause, or because a certain attribute name is not in the table, or is ambiguous (as it happens when a table has two columns with the same name). The formalization of legal SQL expressions has mostly been disregarded in other work, either because the formalized syntax was not sufficiently close to realistic SQL (\HoTTSQL does not use attribute or table names), or because it was decided to assign a dummy semantics to illegal expressions (as in \SQLCoq).

There are distinct judgments for the well-formedness of attribute names and terms, and five distinct, mutually defined judgments for tables, frames, generators, conditions, queries and existentially nested queries. Each judgment mentions a context $\Gamma$  which assigns a schema (list of attribute names) to each table declared in a |FROM| clause. A parameter $D$ (\emph{data base}) provides a partial map from table names $x$ to their (optional) schema $D(x)$.

\begin{figure}
\scriptsize{
\begin{tabular}{c}
\rVS
\multicolumn{1}{l}{
%
\fbox{Variables (\texttt{j\_var})}
\hspace{1cm}
\AxiomC{$x \notin \sigma$}
\UnaryInfC{$x \# \sigma \vdash x$}
\DisplayProof
\hspace{1cm}
\AxiomC{$x \neq y$}
\AxiomC{$\sigma \vdash x$}
\BinaryInfC{$y \# \sigma \vdash x$}
\DisplayProof
}
\\

\\
\rVS
\multicolumn{1}{l}{\fbox{Terms (\texttt{j\_tm}, \texttt{j\_tml})}}
\\
\rVS
\AxiomC{$\phantom{A}$}
\UnaryInfC{$\jt{\Gamma}{\mathbf{k}}$}
\DisplayProof
\hspace{.7cm}
\AxiomC{$\phantom{A}$}
\UnaryInfC{$\jt{\Gamma}{\kwnull}$}
\DisplayProof
\hspace{.7cm}
\AxiomC{$\Gamma(n) = \kwsome~\sigma$}
\AxiomC{$\sigma \vdash x$}
\BinaryInfC{$\jt{\Gamma}{n.x}$}
\DisplayProof
\hspace{.7cm}
\AxiomC{$\forall t \in \vect{t'}: \jt{\Gamma}{t}$}
\UnaryInfC{$\jtl{\Gamma}{\vect{t'}}$}
\DisplayProof
\\
\rVS
\multicolumn{1}{l}{\fbox{Queries (\texttt{j\_query})}}
\\
\rVS
\AxiomC{$\jTl{\Gamma}{G}{\Gamma'}$} \noLine
\UnaryInfC{$\jtl{\Gamma',\Gamma}{\vect{t}}$}
\AxiomC{$\jc{\Gamma',\Gamma}{c}$} \noLine
\UnaryInfC{$\tau = \vect{x}$}
\BinaryInfC{$\jq{\Gamma}{\sfwb{\vect{t : x}}{G}{c}}{\tau}$}
\DisplayProof
\hspace{1cm}
\AxiomC{$\jTl{\Gamma}{G}{\Gamma'}$} \noLine
\UnaryInfC{$\jtl{\Gamma',\Gamma}{\dom{\Gamma'}}$}
\AxiomC{$\jc{\Gamma',\Gamma}{c}$} \noLine
\UnaryInfC{$\tau = \flatten{\Gamma'}$}
\BinaryInfC{$\jq{\Gamma}{\ssfwb{G}{c}}{\tau}$}
\DisplayProof
\\
\rVS
\AxiomC{$\jq{\Gamma}{Q_1}{\sigma}$}
\AxiomC{$\jq{\Gamma}{Q_2}{\sigma}$}
\BinaryInfC{$\jq{\Gamma}{Q_1~\{\kwunion~\vert~\kwinters~\vert~\kwexcept\}~[\kwall]~Q_2}{\sigma}$}
\DisplayProof
\smallskip\\
\rVS
\multicolumn{1}{l}{\fbox{Nested queries (\texttt{j\_inquery})}}
\\
\rVS
\AxiomC{$\jTl{\Gamma}{G}{\Gamma'}$} \noLine
\UnaryInfC{$\jtl{\Gamma',\Gamma}{\vect{t}}$}
\AxiomC{$\jc{\Gamma',\Gamma}{c}$}
\BinaryInfC{$\jiq{\Gamma}{\sfwb{\vect{t : x}}{G}{c}}$}
\DisplayProof
\hspace{1cm}
\AxiomC{$\jTl{\Gamma}{G}{\Gamma'}$}
\AxiomC{$\jc{\Gamma',\Gamma}{c}$} 
\BinaryInfC{$\jiq{\Gamma}{\ssfwb{G}{c}}$}
\DisplayProof
\\
\rVS
\AxiomC{$\jq{\Gamma}{Q_1}{\sigma}$}
\AxiomC{$\jq{\Gamma}{Q_2}{\sigma}$}
\BinaryInfC{$\jiq{\Gamma}{Q_1~\{\kwunion~\vert~\kwinters~\vert~\kwexcept\}~[\kwall]~Q_2}$}
\DisplayProof
\\
\rVS
\multicolumn{1}{l}{\fbox{Tables (\texttt{j\_tb})}}
\\
\rVS
\AxiomC{$D(x) = \kwsome~\sigma$}
\UnaryInfC{$\jT{\Gamma}{\kwtb~x}{\sigma}$}
\DisplayProof
\hspace{1cm}
\AxiomC{$\jq{\Gamma}{Q}{\sigma}$}
\UnaryInfC{$\jT{\Gamma}{\kwquery~Q}{\sigma}$}
\DisplayProof
\\
\rVS
\multicolumn{1}{l}{\fbox{Frames and generators (\texttt{j\_btb}, \texttt{j\_btbl})}}
\\
\rVS
\AxiomC{$\phantom{A}$}
\UnaryInfC{$\jbT{\Gamma}{\snil}{\snil}$}
\DisplayProof
\hspace{1cm}
\AxiomC{$\vert \sigma \vert = \vert \sigma' \vert$} \noLine
\UnaryInfC{$\jT{\Gamma}{T}{\sigma}$}
\AxiomC{$\kwnodup~\sigma'$} \noLine
\UnaryInfC{$\jbT{\Gamma}{\vect{U:\tau}}{\Gamma'}$}
\BinaryInfC{$\jbT{\Gamma}{T:\sigma', \vect{U:\tau}}{\sigma', \Gamma'}$}
\DisplayProof
\\
\rVS
\AxiomC{$\phantom{A}$}
\UnaryInfC{$\jbT{\Gamma}{\fnil}{\fnil}$}
\DisplayProof
\hspace{1cm}
\AxiomC{$\jbT{\Gamma}{G}{\Gamma_1}$}
\AxiomC{$\jbT{\Gamma_1, \Gamma}{\vect{T:\sigma}}{\Gamma_2}$}
\BinaryInfC{$\jbT{\Gamma}{\vect{T:\sigma}, \kwlat~G}{\Gamma_2, \Gamma_1}$}
\DisplayProof
\\
\rVS
\multicolumn{1}{l}{\fbox{Conditions (\texttt{j\_cond})}}
\\
\AxiomC{$\phantom{A}$}
\UnaryInfC{$\jc{\Gamma}{\{\kwtrue \vert \kwfalse\}}$}
\DisplayProof
\hspace{1cm}
\AxiomC{$\jt{\Gamma}{t}$}
\UnaryInfC{$\jc{\Gamma}{t~\kwis~[\kwnot]~\kwnull}$}
\DisplayProof
\hspace{1cm}
\AxiomC{$\jc{\Gamma}{c}$}
\UnaryInfC{$\jc{\Gamma}{c~\kwis~\kwtrue}$}
\DisplayProof
\\
\rVS
\AxiomC{$\vert \vect{t} \vert = n$} 
\AxiomC{$\jtl{\Gamma}{\vect{t}}$}
\BinaryInfC{$\jc{\Gamma}{P^n(\vect{t})}$}
\DisplayProof
\hspace{1cm}
\AxiomC{$\jtl{\Gamma}{\vect{t}}$}
\AxiomC{$\vert \vect{t} \vert = \vert \sigma \vert$} \noLine
\UnaryInfC{$\jq{\Gamma}{Q}{\sigma}$}
\BinaryInfC{$\jc{\Gamma}{\vect{t}~\kwin~Q}$}
\DisplayProof
\hspace{1cm}
\AxiomC{$\jiq{\Gamma}{Q}$}
\UnaryInfC{$\jc{\Gamma}{\kwex~Q}$}
\DisplayProof
\\
\rVS
\AxiomC{$\jc{\Gamma}{c_1}$}
\AxiomC{$\jc{\Gamma}{c_2}$}
\BinaryInfC{$\jc{\Gamma}{c_1~\{\kwand~\vert~\kwor\}~c_2}$}
\DisplayProof
\hspace{1cm}
\AxiomC{$\jc{\Gamma}{c}$}
\UnaryInfC{$\jc{\Gamma}{\kwnot~c}$}
\DisplayProof
\end{tabular}
}
\caption{Well-formed SQL syntax.}\label{fig:SQLjudg}
\end{figure}

We review some of the well-formedness rules. The rules for terms state that constant literals $\mathbf{k}$ and null values are well formed in all contexts. To check whether an attribute reference $n.x$ is well formed (where $n$ is a de~Bruijn index referring to a table and $x$ an attribute name), we first perform a lookup of the $n$-th schema in $\Gamma$: if this returns some schema $\sigma$, and the attribute $x$ is declared in $\sigma$ (with no repetitions), then $n.x$ is well formed. The rules for conditions recursively check that nested subqueries be well-formed and that base predicates $P^n$ be applied to exactly $n$ arguments. 


The well-formedness judgments for queries and tables assign a schema to their main argument. Similarly, well-formed frames of tables are assigned the corresponding sequence of schemas, i.e. a context. The well-formedness judgment for generators uses, recursively, the well-formedness of frames, where each frame added to the generator must be well-formed in a suitably extended context (notice that the last frame is added to the left, contrary to SQL syntax, but coherently with Coq's notation for lists), and finally returns a context obtained by concatenating all the contexts assigned to the individual well-formed frames.

The SQL standard allows well-formed queries to return tables whose schema contains repeated attribute names (e.g. |SELECT A, A, B FROM R|), but requires attribute references in terms to be unambiguous (so that, if the previous query appears as part of a larger one, the attribute name |B| can be used, but |A| cannot). This behaviour is faithfully mimicked in our well-formedness judgments: while well-formed terms are required to only use unambiguous attribute references, the rules for queries do not check that the schema assignment be unambiguous. Furthermore, in a |SELECT *| query that is not contained in an |EXISTS| clause, the star is essentially expanded to the attribute names of the input tables (so that, for example, |SELECT * FROM (SELECT A, A FROM R)| is rejected even though the inner query is accepted, and the ambiguous attribute name |A| is not explicitly referenced).

As an exception, when a |SELECT *| query appears inside an |EXISTS| clause (meaning it is only run for the purpose of checking whether its output is empty or not), SQL considers it well-formed even when the star stands for an ambiguous attribute list. Thus we model this situation as a different well-formedness predicate, with a more relaxed rule for |SELECT *|; furthermore, since the output of an existential subquery is thrown away after checking for non-emptiness, this predicate does not return a schema.

In our formalization, we need to prove weakening only for the term judgment, but not for queries, tables or conditions; weakening for terms is almost painless and only requires us to define a lift function that increments table indices by a given $k$.
Thus, if a term $t$ is well-formed in a context $\Gamma$, then it is also well-formed in an extended context $\Gamma',\Gamma$, provided that we lift it by an amount corresponding to the length of $\Gamma'$.
\begin{lemma}
If $\jt{\Gamma}{t}$, then for all 
$\Gamma'$ we have $\jt{\Gamma',\Gamma}{\mathtt{tm\_lift}~t~\vert \Gamma' \vert}$.
\end{lemma}

%% file: model.tex
\section{$K$-relations as an abstract data type}\label{sec:relations}
We recall the notion of \emph{$K$-relation}, introduced in~\cite{Green07} by Green et al.: for a commutative semi-ring $(K,+,\times,0,1)$ (i.e. $(K,+,0)$ and
$(K,\times,1)$ are commutative monoids,  $\times$ distributes over
$+$, and $0$ annihilates $\times$), a $K$-relation is a \emph{finitely supported} function $R$ of type $T \to K$, where by finitely supported we mean that $R~t \neq 0$ only for finitely many $t : T$. $K$-relations constitute a natural model for databases: for example, if $K = \bN$, $R~t$ can be interpreted as the multiplicity of a tuple $t$ in $R$, and finite-supportedness corresponds to the finiteness of bags. 
In Coq, we can represent $K$-relations as (computable) functions: however, each function must be proved finitely supported separately, cluttering the formalization. To minimize the complication, we model $K$-relations by means of an abstract data type (as opposed to the concrete type of functions); this technique was previously used by one of the authors to formalize binding structures~\cite{Ricciotti15}. 

Just as in the theory of programming languages, an abstract data type for $K$-relations does not provide access to implementation details, but offers a selection of operations (union, difference, cartesian product) that are known to preserve the structural properties of $K$-relations, and in particular finite-supportedness. For the purpose of this work, the ADT we describe is specialized to $\bN$-relations; we fully believe our technique can be adapted to general commutative semi-rings (including the provenance semi-rings that provided the original motivation for $K$-relations), with some adaptations due to the fact that our model needs to support operations, like difference, that are not available in a semi-ring.

Our abstract type of relations is defined by means of the following signature:
%
\begin{coq}[texcl]
  Parameter R : nat -> Type.
  Parameter V : Type.
  Definition T := Vector.t V.
  Parameter memb : forall n, R n -> T n -> nat.            (*$\#(r,t)$*)
  Parameter plus : forall n, R n -> R n -> R n.            (*$\oplus$*)
  Parameter minus: forall n, R n -> R n -> R n.            (*$\backslash$*)
  Parameter inter: forall n, R n -> R n -> R n.            (*$\cap$*)
  Parameter times: forall m n, R m -> R n -> R (m + n).    (*$\times$*)
  Parameter sum  : forall m n, R m -> (T m -> T n) -> R n.  (*$\Sigma$*)
  Parameter rsum  : forall m n, R m -> (T m -> R n) -> R n. (*$\biguplus$*)
  Parameter sel  : forall n, R n -> (T n -> bool) -> R n.   (*$\sigma$*)
  Parameter flat : forall n, R n -> R n.                  (*$\Vert \cdot \Vert$*)
  Parameter supp : forall n, R n -> list (T n).
  Parameter Rnil : forall n, R n.
  Parameter Rone : R 0.
  Parameter Rsingle : forall n, T n -> R n.
\end{coq}
This signature declares a type family $\mathtt{R~n}$ of $\mathtt{n}$-ary relations, and a type $\mathtt{V}$ of data values. The type family $\mathtt{T~n}$ of $\mathtt{n}$-tuples is defined as a vector with base type $\mathtt{V}$. The key difference compared to the concrete approach is that, given a relation $r$ and a tuple $t$, both with the same arity, we obtain the multiplicity of $t$ in $r$ as $\#(r,t)$, where $\#(\cdot,\cdot)$ is an abstract operator; the concrete style $r~t$ is not allowed because the type of $R$ is abstract, i.e. we do not know whether it is implemented as a function or as something else.

We also declare binary operators $\oplus$, $\backslash$, and $\cap$ for the disjoint union, difference, and intersection on $n$-ary bags. The cartesian product $\times$ takes two relations of possibly different arity, say $m$ and $n$, and returns a relation of arity $m+n$. 

The operator $\mathtt{sum~r~f}$, for which we use the notation $\sum_r f$ (or, sometimes, $\sum_{x \leftarrow r} f~x$) represents bag comprehension: it takes a relation $r$ of arity $m$ and a function $f$ from $m$-tuples to $n$-tuples, and builds a new relation of arity $n$ as a disjoint union of all the $f~x$, where $x$ is a tuple in $r$, taken with its multiplicity; note that for such an operation to be well-defined, we need $r$ to be finitely supported. We also provide a more general form of comprehension $\mathtt{rsum~r~g}$, with the notation $\biguplus_r g$ (or, equivalently, $\biguplus_{x \leftarrow r} g~x$) where the function $g$ maps $m$-tuples to $n$-relations: the output of this comprehension will be a new relation of arity $n$ built by taking the disjoint union of all the relations $g~x$, where $x$ is a tuple in $r$, taken with its multiplicity. Again, this operation is well-defined only if $r$ is finitely supported.

Filtering is provided by $\mathtt{sel~r~p}$ (notation: $\sigma_p(r)$), where $p$ is a boolean predicate on tuples: this will return a relation that contains all the tuples of $r$ that satisfy $p$, but not the other ones.

We also want to be able to convert a bag $r$ to a set (i.e.,
0/1-valued bag) $\Vert r \Vert$ containing exactly one copy of each
tuple present in $r$ (regardless of the original
multiplicity). Finally, there is an operator $\mathtt{supp~r}$
returning a list of tuples representing the finite support of $r$.

"Rnil n" identifies the standard empty relation of arity "n", and similarly "Rone" is the standard 0-ary singleton containing exactly one copy of the empty tuple. We also provide "Rsingle n t", or the singleton relation containing the tuple "t" of arity "n", although this can easily be defined in terms of "Rone" and "sum".

In our approach, all the operations on abstract relations mentioned so far are declared but not concretely defined. When ADTs are used for programming, nothing more than the signature of all operations is needed, and indeed this suffices in our case as well if all we are interested in is defining the semantics of SQL in terms of abstract relations. However, proving theorems about this semantics would be impossible if we had no clue about what these operations do: how do we know that $\oplus$ really performs a multiset  union, and $\cap$ an intersection? To make reasoning on abstract relations possible without access to their implementation, we will require that any implementation shall provide some correctness criteria, or proofs that all operations behave as expected. 


The full definition of the correctness criteria for abstract relations as we formalized them in Coq is as follows:
\begin{coq}
  Parameter p_ext : 
    forall n, forall r s : R n, 
    (forall t, memb r t = memb s t) -> r = s.
  Parameter p_fs : 
    forall n, forall r : R n, forall t, 
    memb r t > 0 -> List.In t (supp r).
  Parameter p_fs_r : 
    forall n, forall r : R n, forall t, 
    List.In t (supp r) -> memb r t > 0.
  Parameter p_nodup : 
    forall n, forall r : R n, NoDup (supp r).
  Parameter p_plus : 
    forall n, forall r1 r2 : R n, forall t, 
    memb (plus r1 r2) t = memb r1 t + memb r2 t.
  Parameter p_minus : 
    forall n, forall r1 r2 : R n, forall t, 
    memb (minus r1 r2) t = memb r1 t - memb r2 t.
  Parameter p_inter : 
    forall n, forall r1 r2 : R n, forall t, 
    memb (inter r1 r2) t 
    = min (memb r1 t) (memb r2 t).
  Parameter p_times : 
    forall m n, forall r1 : R m, forall r2 : R n,
    forall t t1 t2, t = Vector.append t1 t2 -> 
    memb (times r1 r2) t = memb r1 t1 * memb r2 t2.
  Parameter p_sum : 
    forall m n, forall r : R m, 
    forall f : T m -> T n, forall t, 
    memb (sum r f) t = list_sum (List.map (memb r) 
      (filter (fun x => T_eqb (f x) t) (supp r))).
  Parameter p_rsum : 
    forall m n, forall r : R m, 
	forall f : T m -> R n, forall t,
    memb (rsum r f) t = list_sum (List.map 
	  (fun t0 => memb r t0 * memb (f t0) t) 
	  (supp r)).
  Parameter p_self : 
    forall n, forall r : R n, forall p t, 
    p t = false -> memb (sel r p) t = 0.
  Parameter p_selt : 
    forall n, forall r : R n, forall p t, 
    p t = true -> memb (sel r p) t = memb r t.
  Definition flatnat := fun n => 
    match n with 0 => 0 | _ => 1 end.
  Parameter p_flat : 
    forall n, forall r : R n, forall t, 
    memb (flat r) t = flatnat (memb r t).
  Parameter p_nil : forall n (t : T n), memb Rnil t = 0.
  Parameter p_one : forall t, memb Rone t = 1.
  Parameter p_single : 
    forall n (t : T n), memb (Rsingle t) t = 1.
  Parameter p_single_neq : 
    forall n (t1 t2 : T n), t1 <> t2 -> memb (Rsingle t1) t2 = 0.
\end{coq}

A first, important property is that relations must be extensional: in other words, any two relations containing the same tuples with the same multiplicities, are equal; this is not true of lists, because two lists containing the same elements in a different order are not equal. Relations should also be finitely supported, and we expect the support not to contain duplicates. 
The properties for the standard $0$-ary relations "Rnil" and "Rone" describe the standard $0$-ary relations, which implicitly employs the fact that the only $0$-tuple is the empty tuple.
The properties for "plus", "minus", "inter" express the behaviour of
disjoint union, difference, and intersection: for instance, a tuple $\#(r \oplus s, t)$ is equal to $\#(r,t) + \#(s,t)$. The behaviour of cartesian products is described as follows: if $r_1$ and $r_2$ are, respectively, an $m$-ary and an $n$-ary relation, and $t$ is an $(m+n)$-tuple, we can split $t$ into an $m$-tuple $t_1$ and an $n$-tuple $t_2$, and $\#(r_1 \times r_2,t) = \#(r_1,t_1) * \#(r_2,t_2)$. The behaviour of filtering ("p_self", "p_selt") depends on whether the filter predicate $p$ is satisfied or not: $\#(\sigma_p(r),t)$ is equal to $\#(r,t )$ if $p~t = \mathtt{true}$, but it is zero otherwise. 

The value of $\#(\Vert r \Vert,t)$ is one if $\#(r,t)$ is greater than zero, or  zero otherwise.
Finally, "p_sum" and "p_rsum" describe the behaviour of bag comprehensions by relating it to the support of the base relation: $\#(\sum_r~f,t)$ is equal to the sum of multiplicities of those elements $x$ of $r$ such that $t = f~x$; this value can be obtained by applying standard list functions to $\mathtt{supp}~r$; $\#(\biguplus_r~g,t)$ is equal to the sum of multiplicities of the elements $x$ of $r$ multiplied by the multiplicities of $t$ in $g~x$.

\subsection{A model of $K$-relations}
The properties of "R" that we have assumed describe a ``na\"ive'' presentation of $K$-relations: they really are nothing more than a list of desiderata, providing no argument (other than common sense) to support their own satisfiability. However, we show that an implementation of "R" (that is, in logical terms, a model of its axioms) can be given within the logic of Coq.

Crucially, our implementation relies on the assumption that the type "V" of values be totally ordered under a relation $\leq_\mathtt{V}$; consequently, tuples of type "T n" are also totally ordered under the corresponding lexicographic order $\leq_\mathtt{T~n}$. We then provide an implementation of "R n" by means of a refinement type:

\begin{coq}
  Definition R := fun n => { l : list (T n) & is_sorted l = true }.
  Definition memb {n} : T n -> R n -> nat 
    := fun A x => List.count_occ (projT1 A) x.
\end{coq}
where "is_sorted l" is a computable predicate returning "true" if and only if "l" is sorted according to the order $\leq_\mathtt{T~n}$.
The inhabitants of "R n" are dependent pairs $\langle l,H \rangle$, such that $l : \mathtt{T~n}$ and $H : \mathtt{is\_sorted}~l = \mathtt{true}$. The multiplicity function for relations "memb" is implemented by counting the number of occurrences of a tuple in the sorted list ("count_occ" is a Coq standard library function on lists).

The most important property that this definition must satisfy is extensionality. For any two sorted lists $l_1, l_2$ of the same type, we can indeed prove that whenever they contain the same number of occurrences of all elements, they must be equal: however, to show that $\langle l_1, H_1 \rangle = \langle l_2, H_2 \rangle$ (where $H_i : \mathtt{is\_sorted}~l_i = \mathtt{true}$) we also need to know that the two proofs $H_1$ and $H_2$ are equal. Knowing that $l_1 = l_2$, this is a consequence of uniqueness of identity proofs (UIP) on "bool", which is provable in Coq (unlike generalized UIP).

Operations on relations can often be implemented using the following scheme:
\begin{coq}
  Definition op {n} : R n -> R n -> ... -> R n := fun A B ... => 
    existT _ (sort (f (projT1 A) (projT1 B) ...)) (sort_is_sorted _).
\end{coq}
where "f" is some function of type "list (T n)" $\to$ "list (T n)" $\to \ldots \to$ "list (T n)". Given relations "A", "B" ... we apply "f" to the underlying lists "projT1 A", "projT1 B",...; then, we sort the result and we lift it to a relation by means of the dependent pair constructor "existT". The theorem "sort_is_sorted" states that
"is_sorted (sort l) = true" for all lists "l". The scheme is used to define disjoint union, difference and intersection:
\begin{coq}
  Definition plus {n} : R n -> R n -> R n 
    := fun A B => existT _ 
    (sort (projT1 A++projT1 B)) (sort_is_sorted _).
  Definition minus {n} : R n -> R n -> R n
    := fun A B => existT _ 
      (sort (list_minus (projT1 A) (projT1 B))) 
      (sort_is_sorted _).
  Definition inter {n} : R n -> R n -> R n
    := fun A B => existT _ 
      (sort (list_inter (projT1 A) (projT1 B))) 
      (sort_is_sorted _).
\end{coq}
For disjoint union, "f" is just list concatenation. For difference, we have to provide a function "list_minus", which could be defined directly by recursion in the obvious way; instead, we decided to use the following definition:
\begin{coq}
  Definition list_minus {n} : list (T n) -> list (T n) -> list (T n)
  := fun l1 l2 => let l := nodup _ (l1 ++ l2) in
   List.fold_left (fun acc x => 
    acc ++ repeat x (count_occ _ l1 x - count_occ _ l2 x)) l List.nil.
\end{coq}
This definition first builds a duplicate-free list "l" containing all
tuples that may be required to appear in the output. Then, for each tuple "x" in "l", we add to the output as many copies of "x" as required (this is the difference between the number of occurrences of "x" in "l1" and "l2"). The advantage of this definition is that it is explicitly based on the correctness property of relational difference: thus, the proof of correctness is somewhat more direct. The same approach can be used for intersection and, with adaptations, for cartesian product.
Finally, "sum", "rsum", "sel", and "flat" reflect respectively list map, concat-map, filter, and duplicate elimination.

We do not provide an operation to test for the emptiness of a relation, or to compute the number of tuples in a relation; however, this may be readily expressed by means of "sum": all we need to do is map all tuples to the same distinguished tuple. The simplest option is to use the empty tuple $\langle\rangle$ and check for membership:
\[
\mathtt{card}~S := \#(\sum_S (\lambda x.\langle\rangle), \langle\rangle)
\]
The correctness criterion for "card", stating that the cardinality of a relation is equal to the sum of the number of occurrences of all tuples in its support, is an immediate consequence of its definition and of the property "p_sum":
\begin{lemma}
$\mathtt{card}~S = \mathtt{list\_sum}~[ \#(S,x) \vert x \leftarrow \mathtt{supp}~S ]$
\end{lemma}

%% file: semantics.tex
\section{Formalized semantics}\label{sec:semantics}
The formal semantics of SQL can be given as a recursively defined function or as an inductive judgment. Although in our development we considered both options and performed some of the proofs in both styles, we will here only discuss the latter, which has proven considerably easier to reason on.
As we intend to prove that three-valued logic (3VL) does not add expressive power to SQL compared to Boolean (two-valued) logic (2VL), we actually need two different definitions: a semantic evaluation based on 3VL (corresponding to the SQL standard), and a similar evaluation based on Boolean logic. We factorized the two definitions, which can be obtained by instantiating a Coq functor to the chosen notion of truth value.

\subsection{Truth values}
For the semantics of SQL conditions, we use an abstract type $\bfB$ of truth values: this can be instantiated to Boolean values ("bool") or to 3VL values ("tribool", with values "ttrue" or $\mathrm{T}$, "tfalse" or $\mathrm{F}$, and "unknown" or $\mathrm{U}$): in the latter case, we obtain the usual three-valued logic of SQL. Technically, 3VL refers either to Kleene's ``strong logic of indeterminacy'', or to Łukasiewicz's L3 logic, which share the same values and truth tables for conjunction, disjunction, and negation (Figure~\ref{fig:3VLops}); both logics also define an implication connective, with different truth tables: since implication plays no role in the semantics of SQL, it is omitted in our formalization.

\begin{figure}
\begin{tabular}{|c|c|c|c|}
\hline
$\land$ & $\mathbf{F}$ & $\mathbf{U}$ & $\mathbf{T}$ \\
\hline
$\mathbf{F}$ & $\mathrm{F}$ & $\mathrm{F}$ & $\mathrm{F}$ \\
\hline
$\mathbf{U}$ & $\mathrm{F}$ & $\mathrm{U}$ & $\mathrm{U}$ \\
\hline
$\mathbf{T}$ & $\mathrm{F}$ & $\mathrm{U}$ & $\mathrm{T}$ \\
\hline
\end{tabular}
\hspace{2cm}
\begin{tabular}{|c|c|c|c|}
\hline
$\lor$ & $\mathbf{F}$ & $\mathbf{U}$ & $\mathbf{T}$ \\
\hline
$\mathbf{F}$ & $\mathrm{F}$ & $\mathrm{U}$ & $\mathrm{T}$ \\
\hline
$\mathbf{U}$ & $\mathrm{U}$ & $\mathrm{U}$ & $\mathrm{T}$ \\
\hline
$\mathbf{T}$ & $\mathrm{T}$ & $\mathrm{T}$ & $\mathrm{T}$ \\
\hline
\end{tabular}
\hspace{2cm}
\begin{tabular}{|c|c|}
\hline
$A$ & $\lnot A$ \\
\hline
$\mathbf{F}$ & $\mathrm{T}$ \\
\hline
$\mathbf{U}$ & $\mathrm{U}$ \\
\hline
$\mathbf{T}$ & $\mathrm{F}$ \\
\hline
\end{tabular}
\caption{Three-valued logic truth tables.}\label{fig:3VLops}
\end{figure}

For convenience, "bool" and "tribool" will be packaged in modules "Sem2" and "Sem3" of type "SEM" together with some of their properties.
\begin{coq}
Module Type SEM (Db : DB).
  Import Db.
  Parameter B         : Type.
  Parameter btrue     : B.
  Parameter bfalse    : B.
  Parameter bmaybe    : B.
  Parameter band      : B -> B -> B.
  Parameter bor       : B -> B -> B.
  Parameter bneg      : B -> B.
  Parameter is_btrue  : B -> bool.
  Parameter is_bfalse : B -> bool.
  Parameter of_bool   : bool -> B.
  Parameter veq       : Value -> Value -> B.

  Hypothesis sem_bpred : forall n, 
   (forall l : list BaseConst, length l = n -> bool) 
    -> forall l : list Value, length l = n -> B.
End SEM.
\end{coq}
"SEM" declares the abstract truth values "btrue", "bfalse", "bmaybe" (in "Sem3", "bmaybe" is mapped to the uncertain value "unknown"; in "Sem2", both "bmaybe" and "bfalse" are mapped to "false"). "SEM" also declares abstract operations ("band", "bor", "bneg"), operations relating abstract truth values and Booleans ("is_btrue", "is_bfalse", "of_bool"), a "B"-valued equality predicate for SQL values (including "NULL"s), and an operation "sem_bpred" which lifts $n$-ary Boolean-valued predicates on constants to "B"-valued predicates on SQL values (including "NULL"s): this is used to define the semantics of SQL conditions using base predicates. A theorem "sem_bpred_elim" describes the behaviour of "sem_bpred": if the list of values "l" provided as input does not contain "NULL"s, it is converted to a list of constants "cl", then the base predicate "p" is applied to "cl"; this yields a Boolean value that is converted to "B" by means of "of_bool". If "l" contains one or more "NULL"s, "sem_bpred" will return "bmaybe".

\subsection{A functor of SQL semantics}
In Coq, when defining a collection of partial maps for expressions subject to well-formedness conditions, we can use an ``algorithmic approach'' based on dependently typed functions, or a ``declarative approach'' based on inductively defined judgments. The two alternatives come both with benefits and drawbacks; for the purposes of this formalization, consisting of dozens of cases with non-trivial definitions, we judged the declarative approach as more suitable, as it helps decouple proof obligations from definitions.
Our inductive judgments implement SQL semantics according to the following style. When a certain expression (query, table or condition) is well-formed for a context $\Gamma$, we expect its semantics to depend on the value assignments for the variables declared in $\Gamma$: we call such an assignment an \emph{environment} for $\Gamma$ (which has type $\mathtt{env}~\Gamma$ in our formalization); thus, we define a semantics that assigns to each well-formed expression an \emph{evaluation}, i.e. a function taking as input an environment, and returning as output a value, tuple, relation, or truth value. Subsequent proofs do not rely on the concrete structure of environments, but internally they are represented as lists of lists of values, which have to match the structure of $\Gamma$:

\begin{coq}
  Definition preenv := list (list Value).
  Definition env := fun g => { h : preenv &
   List.map (@List.length Name) g = List.map (@List.length Value) h }.
\end{coq}

Similarly to well-formedness judgments, we have judgments for the semantics of attribute names and terms, and five mutually defined judgments for the various expression types of SQL. Figure~\ref{fig:semtypes} summarizes the judgments, highlighting the type of the evaluation they return. In our notation, we use judgments $\gsem{\mathfrak{J}}$ with a superscript $\bfB$ denoting their definition can be instantiated to different notions of truth value, in particular, "bool" and "tribool"; we will use the notation $\bsem{\mathfrak{J}}$ and $\tsem{\mathfrak{J}}$ for the two instances. The semantics of attributes and terms does not depend on the notion of truth value, thus the corresponding judgments do not have a superscript. Concretely, our Coq formalization provides a module "Evl" for the judgments that do not depend on $\bfB$, and a functor "SQLSemantics" for the other judgments, which we instantiate with the "Sem2" and "Sem3" we described in the previous section.

We can prove that our semantics assigns only one evaluation to each SQL expression.
\begin{lemma}\label{lem:funeval}
For all judgments $\mathfrak{J}$, if $\gsem{\mathfrak{J}} \eval S$ and $\gsem{\mathfrak{J}} \eval S'$, then $S=S'$. 
\end{lemma}
Thanks to the previous result, whenever $\sem{\mathfrak{J}} \eval S$, we are allowed to use the notation $\sem{\mathfrak{J}}$ for the semantic evaluation $S$, with no ambiguity.
\begin{figure}[t]
\footnotesize{
\[
\begin{array}{llll}
\mbox{Simple attributes} & \sem{\tau \vdash x} \eval S_x & \mbox{s.t. } S_x & : \mathtt{env}~[\tau] \to \mathtt{V}
\\
\mbox{Full attributes} & \sem{\Gamma \vdash n.x} \eval S_{n.x} & \mbox{s.t. } S_{n.x} & : \mathtt{env}~\Gamma \to \mathtt{V}
\\
\mbox{Terms} & \sem{\Gamma \vdash t} \eval S_{t} & \mbox{s.t. } S_t & : \mathtt{env}~\Gamma \to \mathtt{V}
\\
& \sem{\Gamma \vdash \vect{t}} \eval S_{\vect{t}} & \mbox{s.t. } S_{\vect{t}} & : \mathtt{env}~\Gamma \to \mathtt{T}~\len{\vect{t}}
\\
\mbox{Queries} & \gsem{\jq{\Gamma}{Q}{\tau}} \eval S_Q & \mbox{s.t. } S_Q & : \mathtt{env}~\Gamma \to \mathtt{R}~\len{\tau}
\\
\mbox{Nested queries} & \gsem{\jiq{\Gamma}{Q}} \eval S_Q & \mbox{s.t. } S_Q & : \mathtt{env}~\Gamma \to \mathtt{bool}
\\
\mbox{Tables} & \gsem{\jT{\Gamma}{T}{\tau}} \eval S_T & \mbox{s.t. } S_T & : \mathtt{env}~\Gamma \to \mathtt{R}~\len{\tau}
\\
\mbox{Frames} & \gsem{\jTl{\Gamma}{\vect{T:\tau}}{\Gamma'}} \eval S_{\vect{T}} & \mbox{s.t. } S_{\vect{T}} & : \mathtt{env}~\Gamma \to \mathtt{R}~\len{\mathtt{concat}~\Gamma'}
\\
\mbox{Generators} & \gsem{\jTl{\Gamma}{G}{\Gamma'}} \eval S_{G} & \mbox{s.t. } S_{G} & : \mathtt{env}~\Gamma \to \mathtt{R}~\len{\mathtt{concat}~\Gamma'}
\\
\mbox{Conditions} & \gsem{\jc{\Gamma}{c}} \eval S_c & \mbox{s.t. } S_c & : \mathtt{env}~\Gamma \to \bfB
\end{array}
\]
}
\caption{Formal semantics of SQL (types).}\label{fig:semtypes}
\end{figure}
Simple attributes are defined in a schema rather than a context: their semantics $\sem{\tau \vdash x}$ maps an environment for the singleton context $[\tau]$ to a value. Similarly, the semantics of fully qualified attributes $\sem{\Gamma \vdash n.x}$ maps an environment for $\Gamma$ to a value. In both cases, the output value is obtained by lookup into the environment.

The evaluation of terms $\sem{\jt{\Gamma}{t}}$ returns a value for $t$ given a certain environment $\gamma$ for $\Gamma$. In our definition, terms can be either full attributes $n.x$, constants $\bfk$, or "NULL". We have just explained the semantics of full attributes; on the other hand, constants and "NULL"s are already values and can thus be returned as such. The evaluation of term sequences $\sem{\Gamma \vdash \vect{t}}$, given an environment, returns the tuple of values corresponding to each of the terms and is implemented in the obvious way.

Queries and tables ($\gsem{\jq{\Gamma}{Q}{\tau}}$, $\gsem{\jT{\Gamma}{T}{\tau}}$) evaluate to relations whose arity corresponds to the length of their schema $\tau$ (written $\vert \tau \vert$). Existential subqueries evaluate to a non-emptiness test: their evaluation returns a Boolean which is "true" if, and only if, the query returns a non-empty relation. The evaluation of frames $\gsem{\jTl{\Gamma}{\vect{(T:\tau)}}{\Gamma'}}$ returns again a relation, whose arity corresponds to the arity of their cross join: this is obtained by flattening $\Gamma'$ and counting its elements; the judgment for generators operates in a similar way. Conditions evaluate to truth values in $\bfB$: in particular, the evaluation of logical values and connectives |TRUE|, |FALSE|, |AND|, |OR| and |NOT| exploits the operations "btrue", "bfalse", "band", "bor", and "bneg" provided to the functor by the input module "SEM"; similarly, atomic predicates are evaluated using the operation "sem_bpred", while to evaluate $c~\kwis~\kwtrue$, we first evaluate the condition $c$ recursively, obtaining a truth value in $\bfB$, then we pass this value to "is_btrue", which returns a "bool" (even when we are using 3VL), and finally coerce it back to $\bfB$ using the operation "of_bool" (this construction ensures that |IS TRUE| always evaluates to either "btrue" or "bfalse").

As for well-formedness judgments, we prove a weakening lemma:

\begin{lemma}
If $\gsem{\jt{\Gamma}{t}} \eval S$ then, for all $\Gamma'$, we have
\[
\gsem{\jt{\Gamma',\Gamma}{\mathtt{tm\_lift}~t~\vert \Gamma' \vert}} \eval \lambda\eta.\mathtt{subenv2}~\eta
\] 
where $\mathtt{subenv2} : \mathtt{env}~(\Gamma',\Gamma) \to \mathtt{env}~\Gamma$ takes an environment for a context obtained by concatenation and returns its right projection.
\end{lemma}

\subsection{Discussion}
To explain the semantics of queries, let us consider the informal definition~\cite{guagliardo17}:
\[
\footnotesize{
\begin{array}{l}
\sem{\sfwc{\vect{t : x}}{\vect{T : \sigma}}{c}}~\eta =
\quad \left\{ k \cdot \sem{t} \eta' \middle\vert \#\left(\sem{\vect{T : \sigma}}~\eta, \vect{V} \right) = k, \sem{c}~\eta' = \kwtt \right\}
\end{array}
}
\]
where $\eta'$ is defined as the extension of evaluation $\eta$ assigning values $\vect{V}$ to fully qualified attributes from $\vect{T : \sigma}$ (in the notation used by~\cite{guagliardo17}, $\eta' := \eta \stackrel{\vect{V}}{\oplus} \ell(\vect{T : \sigma})$). This definition operates by taking the semantics of the tables in the $\kwfrom$ clause (their cartesian product). For each tuple $\vect{V}$ contained $k$ times in this multiset, we extend the environment $\eta$ with $\vect{V}$, obtaining $\eta'$. If $c$ evaluates to $\kwtt$ in the extended environment, we yield $k$ copies of $\sem{t}~\eta'$ in the result.

The definition above makes implicit assumptions (particularly, the fact that $\eta$ and $\eta'$ should be good environments for the expressions whose semantics is evaluated), and at the same time introduces a certain redundancy by computing the number $k$ of occurrences of $\vect{V}$ in the input tables, and using it to yield the same number of copies of output tuples.

In our formalization, the semantics above is implemented using abstract relations rather than multisets. While in the paper definition the environment $\eta'$ is obtained by shadowing names already defined in $\eta$, we can dispense with that since we rule out name clashes syntactically, thanks to the use of de~Bruijn indices. The implementation uses dependent types and some of the rules use equality proofs to allow premises and conclusions to typecheck: we will not describe these technical details here, and refer the interested reader to the Coq scripts.

\begin{minipage}{0.9\textwidth}
\begin{center}
\footnotesize{
\AxiomC{$\sem{\jTl{\Gamma}{G}{\Gamma'}} \eval S_{G}$}
\AxiomC{$\sem{\jc{\Gamma',\Gamma}{c}} \eval S_c$}
\AxiomC{$\sem{\jtl{\Gamma',\Gamma}{\vect{t}}} \eval S_{\vect{t}}$}
\TrinaryInfC{$
\begin{array}{l}
  \sem{\jq{\Gamma}{\sfwc{\vect{t : x}}{G}{c}}{\sigma'}} \eval \lambda \eta. \\
  \qquad \kwlet~p := \lambda \vect{v}. \mathtt{is\_btrue}~(S_c~([\Gamma' \mapsto \vect{v}] \app \eta))~\kwlin \\
  \qquad \kwlet~R := \sigma_p(S_G)~\eta)~\kwlin \\
  \qquad \kwlet~f := \lambda \vect{v}. S_{\vect{t}}~([\Gamma' \mapsto \vect{v}]\app \eta)~\kwlin~ \displaystyle \sum_{R}~f
  \end{array}$}
\DisplayProof
}
\end{center}
\end{minipage}

In this mechanized version, the input to the |SELECT| is generalized to one that may include lateral joins, by using $G = \vect{T_1:\sigma_1}~\kwlat~\ldots~\kwlat~\vect{T_n:\sigma_n}$ (we get the original version for $n=1$); the relation $R := \sigma_p(S_{G}~\eta)$ replaces the predicate in the multiset comprehension, whereas $f$ assumes the role of the output expression. Whenever a certain tuple $\vect{V}$ appears $k$ times in $R$, the relational comprehension operator adds $f~V$ to the output the same number of times, so it is unnecessary to make $k$ explicit in the definition. The operation $[\Gamma' \mapsto \vect{v}]$ creates an environment for $\Gamma'$ by providing a tuple $\vect{v}$ of correct length: this constitutes a proof obligation that can be fulfilled by noticing that each $\vect{v}$ ultimately comes from $\sem{\jTl{\Gamma}{G}{\Gamma'}}$, whose type is $\mathtt{env}~\Gamma \to \mathtt{R}~\len{\mathtt{concat}~\Gamma'}$. Since $G$ represents a telescope of lateral joins, its semantics deserves some attention. The interesting case is the following:

\begin{minipage}{0.9\textwidth}
\begin{center}
\footnotesize{
\AxiomC{$\sem{\jTl{\Gamma}{\vect{T:\sigma}}{\Gamma'}} \eval S_{\vect{T}}$}
\AxiomC{$\sem{\jTl{\Gamma',\Gamma}{G}{\Gamma''}} \eval S_G$}
\BinaryInfC{$
\begin{array}{l}
  \sem{\jTl{\Gamma}{\vect{T : \sigma}~\kwlat~G}{\Gamma'',\Gamma'}} \eval \lambda \eta. \\
  \qquad \kwlet~R := (S_{\vect{T}}~\eta)~\kwlin \\
  \qquad \kwlet~f := \lambda \vect{v}.(S_G~([\Gamma' \mapsto \vect{v}]\app \eta)) \times (\mathtt{R\_single}~\vect{v})~\kwlin~ \displaystyle \biguplus_{R}~f
  \end{array}$}
\DisplayProof
}
\end{center}
\end{minipage}

To evaluate a generator $\vect{T:\sigma}~\kwlat~G$ given an environment $\eta$, we first evaluate $\vect{T:\sigma}$ in $\eta$, obtaining a relation $R$; then, for each tuple $\vect{v}$ in $R$, we extend $\eta$ with that particular value of $\vect{v}$ and evaluate $G$ recursively in it; we take the product of the resulting relation with the singleton containing the tuple $\vect{v}$; finally, we perform a disjoint union for all the $\vect{v}$. Notice that in the absence of $\kwlat$ it would have sufficed to perform a product between the semantics of $\vect{T:\sigma}$ and that of $G$; that is not possible here, because we need to consider a different semantics of $G$ for each element of the semantics of $\vect{T:\sigma}$.

Perhaps a more intuitive way of implementing this semantics would have been a judgment in the form $\sem{\jq{\Gamma}{Q}{\tau}}~\eta \eval R$, where $\eta$ is an environment for $\Gamma$ and $R$ is the relation resulting from the evaluation of $Q$ in that specific environment; however, in the example above, we can see that, in order to compute the relation resulting from the evaluation of the query, the predicate $p$ is used to evaluate the condition $c$ in various different environments: this forces us to evaluate conditions to functions taking as input an environment, and due to the mutual definition of conditions and queries, the evaluation of queries must result in a function as well.

The appendix contains the full definition of the semantics we formalized. We only consider here the judgment used to evaluate |IN| conditions, as it deserves a brief explanation:
\begin{minipage}{\textwidth}
\footnotesize{
\AxiomC{$\sem{\jtl{\Gamma}{\vect{t}}} \eval S_{\vect{t}}$}
\AxiomC{$\sem{\jq{\Gamma}{Q}{\tau}} \eval S_Q$}
\BinaryInfC{$
\begin{array}{l}
 \sem{\jc{\Gamma}{\vect{t}~\kwin~Q}} \eval \lambda \eta. \\
 \qquad \kwlet~n^\kwtt := \kwcard~(\sigma_{p^\kwtt}(S_Q~\eta))~\kwlin \\
 \qquad \kwlet~n^\kwuu := \kwcard~(\sigma_{p^\kwuu}(S_Q~\eta))~\kwlin \\
 \qquad \kwif~(0 < n^\kwtt)~\kwthen~\mathtt{btrue} \\
 \qquad \kwelse~\kwif~(0 < n^\kwuu)~\kwthen~\mathtt{bmaybe} \\
 \qquad \kwelse~\mathtt{bfalse}
\end{array}
$}
}
\end{minipage}
The membership condition must bridge the gap between the three-valued logic of SQL and the Boolean logic used by abstract relations: in particular, to check whether a tuple $\vect{t}$ appears in the result of a query $Q$, we cannot simply evaluate $\vect{t}$ to $\vect{V}$ and $Q$ to $S$ and check whether $\#(S,\vect{V})$ is greater than zero, because in three-valued logic "NULL" is not equal to itself. Instead, given the semantics of $Q$, we compute the number $n^\kwtt$ of tuples that are equal to $\vect{V}$ and the number $n^\kwuu$ of the tuples of $S$ that are not different from $\vect{V}$ (i.e. the matching is up to the presence of some "NULL"s). If $n^\kwtt$ is greater than zero, then the condition evaluates to "btrue"; if $n^ \kwtt = 0$ but $n^\kwuu > 0$, the condition evaluates to "bmaybe"; if both values are zero, then the tuple is certainly not in the result of $Q$ and the condition evaluates to "bfalse".

The predicates $p^\kwtt$ and $p^\kwuu$ used in the definition are defined as follows:
\begin{align*}
p^\kwtt & := \lambda \vect{V}. \mathtt{fold\_right2}~(\lambda v,w,\mathtt{acc}. \mathtt{acc} \land \mathtt{is\_btrue}~(\mathtt{veq}~v~w))~\mathtt{true}~\vect{V}~(S_Q~\eta)
\\
p^\kwuu & := \lambda \vect{V}. \mathtt{fold\_right2}~(\lambda v,w,\mathtt{acc}. \mathtt{acc} \land \lnot \mathtt{is\_bfalse}~(\mathtt{veq}~v~w))~\mathtt{true}~\vect{V}~(S_Q~\eta)
\end{align*}
Value equality "veq : V ->   V ->  B" returns "bmaybe" when either of the two arguments is "NULL", otherwise corresponds to syntactic equality: "fold_right2" iterates "veq" on pairs of values from the two tuples $\vect{V}$ and $S_Q~\eta$. Although in Boolean logic a predicate is true precisely when it is not false, in "tribool" the $p^\kwtt$ and $p^\kwuu$ may assume different values.

%% file: rewrite.tex
\section{Validation of rewrite rules}\label{sec:rewrite}
Now that we have a formalized semantics of \SQLNull, it is a good time to show that it can be used to verify the soundness of some rewrite rules. The two rules we consider allow tables in the |FROM| clause of a query to be shuffled, and nested queries to be unnested. In the following statements,  given an index $n$ and schema $\sigma = x_1,\ldots,x_k$, we will write $n.\sigma$ as a shorthand for the term sequence $n.x_1,\ldots,n.x_k$; if $\vect{u} = u_1,\ldots,u_k$, we will write $\subst{n.\sigma}{\vect{u}}$ for the simultaneous substitution of $u_i$ for $x_i$, where $i = 1,\ldots,k$. The symbol $\simeq$ represents heterogeneous equality.

\begin{theorem}
Let $\len{\tau'} = \len{\sigma_1} + \len{\sigma_2}$, and $S, S'$ evaluations such that
\[
\footnotesize{
\begin{array}{l}
\sem{
 \jq[]{\Gamma}{
 \kwsel~\ast~
 \kwfrom~T_1 : \sigma_1, T_2 : \sigma_2
 }{\tau}
 } \eval S
\\
\sem{
 \jq[]{\Gamma}{
  \kwsel~(1.\sigma_1, 0.\sigma_2):\tau'~
  \kwfrom~T_2:\sigma_2,T_1:\sigma_1
  }{\tau'}
 } \eval S'
\end{array}
}
\]
Then for all $\eta : \mathtt{env}~\Gamma$, we have $S~\eta \simeq S'~\eta$.
\end{theorem}
\begin{proof}
The proof proceeds by inversion on the derivation of the two semantic judgments; the hypothesis on the length of $\tau'$ is required for the select clause of the second query to be adequate. The goal simplifies to:
\[
\footnotesize{
\#\left( \sum_{\vect{v} \leftarrow S_{\kwfrom}~\eta} \vect{v}, r_1 \right) \simeq \#\left(\sum_{\vect{v} \leftarrow S'_{\kwfrom}~\eta} (S'_{\kwsel}~([\Gamma'' \mapsto \vect{v}]\app\eta)),r_2 \right)
}
\]
under the hypotheses $r_1 \simeq r_2$, $\sem{\jbT{\Gamma}{T_1:\sigma_1,T_2:\sigma_2}{\Gamma'}} \eval S_{\kwfrom}$, 
\linebreak
$\sem{\jbT{\Gamma}{T_2:\sigma_2,T_1:\sigma_1}{\Gamma''}} \eval S'_{\kwfrom}$, $\sem{\jtl{\Gamma'',\Gamma}{1.\sigma_1,0.\sigma_2}} \eval S'_{\kwsel}$. We prove by functional extensionality that the rhs is equal to $\#(\sum_{\vect{v} \leftarrow S'_{\kwfrom}~\eta} (\mathit{flip}~\vect{v},r_2)$, 
where $\mathit{flip}$ is the function that takes a vector of length $\len{\sigma_2}+\len{\sigma_1}$ and swaps the first $\len{\sigma_2}$ elements with the last $\len{\sigma_1}$. Then the goal becomes $\#(S_{\kwfrom},r_1) = \#(S'_{\kwfrom},\mathit{flip}~r_2)$,
which is easily obtained by inversion on 
$S_{\kwfrom}$ and $S'_{\kwfrom}$.
\end{proof}

\begin{theorem}
Let $S, S'$ be evaluations such that
\[
\footnotesize{
\begin{array}{l}
\sem{
 \jq[]{\Gamma}{
 \kwsel~\vect{t:x}~\kwfrom~\kwquery~
 (\kwsel~\vect{u:y}~\kwfrom~T:\sigma_2~\kwwhere~c) : \sigma_1
 }{\tau}
 } \eval S
\\
\sem{
 \jq[]{\Gamma}{
  \kwsel~(\vect{t:x})\subst{0.\sigma_1}{\vect{u}}~
  \kwfrom~T:\sigma_2~\kwwhere~c
 }{\tau'}
 } \eval S'
\end{array}
}
\]
Then for all $\eta : \mathtt{env}~\Gamma$, we have $S~\eta \simeq S'~\eta$.
\end{theorem}
\begin{proof}
By inversion on the derivation of the two evaluations (and also using
\linebreak
Lemma~\ref{lem:funeval}), we know that 
$\sem{\jT{\Gamma}{T}{\sigma_2}} \eval S_{\kwfrom}$,
$\sem{\jtl{\sigma_1,\Gamma}{\vect{t}}} \eval S_{\kwsel}$,
\linebreak
$\sem{\jtl{\sigma_2,\Gamma}{\vect{u}}} \eval S'_{\kwsel}$,
$\sem{\jc{\sigma_2,\Gamma}{c}} \eval S_c$,
$\sem{\jtl{\sigma_2,\Gamma}{(\vect{t:x})\subst{0.\sigma_1}{\vect{u}}}} \eval S''_{\kwsel}$.
\linebreak
The lhs of the thesis computes to an abstract expression containing two nested $\sum$ operations; we prove the general result that $\sum_{\sum_r~f}~g = \sum_r~(g \circ f)$ and obtain the new lhs:
\[
\footnotesize{
  \sum_{\vect{w} \leftarrow 
    \sigma_{p_c} (S_{\kwfrom}~\eta)} 
	(S_{\kwsel} ([\sigma_1 \mapsto (S'_{\kwsel} ([\sigma_2 \mapsto \vect{w}]\app\eta))]\app\eta))
}
\]
where $p_c(\vect{w}) ) := S_c~([\sigma_2 \mapsto \vect{w}]\app\eta)$. The rhs of the goal computes to:
\[
\footnotesize{
\sum_{\vect{w} \leftarrow
  \sigma_{p_c} (S_{\kwfrom}~\eta)}
  (S''_{\kwsel}~([\sigma_2 \mapsto \vect{w}]\app\eta))
}
\]
Then, for the lhs and rhs to be equal, we only need to prove the following:
\[
\footnotesize{
(S_{\kwsel} ([\sigma_1 \mapsto (S'_{\kwsel} ([\sigma_2 \mapsto \vect{w}]\app\eta))]\app\eta))
\simeq (S''_{\kwsel}~([\sigma_2 \mapsto \vect{w}]\app\eta))
}
\]
This is a property of substitution that we prove by induction on the sequence of terms $\vect{t}$.
\end{proof}

%% file: tribool.tex
\section{Elimination of three-valued logic}\label{sec:translation}
We now move to formalizing Guagliardo and Libkin's proof that SQL has the same expressive power under Boolean and three-valued logic, in the sense that for every query evaluated under 3VL, there exists another query with the same semantics in Boolean logic, and vice-versa. The proof is constructive: we exhibit an (algorithmic) transformation $(\cdot)^\kwtt$ which turns a query for 3VL-SQL into Boolean-SQL (a much simpler transformation $(\cdot)^*$ operates in the opposite direction). The transformation $(\cdot)^\kwtt$ is defined by mutual recursion on queries, tables, and conditions; more precisely, $(\cdot)^\kwtt$ is mutually defined with an auxiliary transformation $(\cdot)^\kwff$, operating on conditions only: the rationale is that while $c^\kwtt$ is "true" in Boolean logic when $c$ is "ttrue" in 3VL, $c^\kwff$ is "true" in Boolean logic when $c$ is "tfalse" in 3VL; as a corollary, when $c$ evaluates to 3VL "unknown", both $c^\kwtt$ and $c^\kwff$ are Boolean "false".  

\begin{figure}[!h]
\scriptsize{
\[
\begin{array}{ccc}
\begin{array}{rl}
\kwtrue^\kwtt & = \kwtrue 
\\
\kwfalse^\kwtt & = \kwfalse 
\\
(t~\kwis~\kwnull)^\kwtt & = t~\kwis~\kwnull 
\\
(t~\kwis~\kwnot~\kwnull)^\kwtt & = t~\kwis~\kwnot~\kwnull 
\\
(c~\kwis~\kwtrue)^\kwtt & = c^\kwtt
\\
P^n(\vect{t})^\kwtt & = P^n(\vect{t})
\\
(\kwex~Q)^\kwtt & = \kwex~Q^\kwtt
\\
(c_1~\kwand~c_2)^\kwtt & = c_1^\kwtt~\kwand~c_2^\kwtt
\\
(c_1~\kwor~c_2)^\kwtt & = c_1^\kwtt~\kwor~c_2^\kwtt
\\
(\kwnot~c)^\kwtt & = c^\kwff
\end{array}
& \quad &
\begin{array}{rl}
\kwtrue^\kwff & = \kwfalse
\\
\kwfalse^\kwff & = \kwtrue \\
(t~\kwis~\kwnull)^\kwff & = t~\kwis~\kwnot~\kwnull 
\\
(t~\kwis~\kwnot~\kwnull)^\kwff & = t~\kwis~\kwnull 
\\
(c~\kwis~\kwtrue)^\kwff & = \kwnot~c^\kwtt
\\
P^n(\vect{t})^\kwff & = \kwnot~P^n(\vect{t})~\kwand~\vect{t~\kwis~\kwnot~\kwnull}
\\
(\kwex~Q)^\kwff & = \kwnot~\kwex~Q^\kwtt
\\
(c_1~\kwand~c_2)^\kwff & = c_1^\kwff~\kwor~c_2^\kwff
\\
(c_1~\kwor~c_2)^\kwff & = c_1^\kwff~\kwand~c_2^\kwff
\\
(\kwnot~c)^\kwff & = c^\kwtt
\end{array}
\\
\phantom{nothing}
\\
\\
\multicolumn{3}{l}{
 (\vect{t}~\kwin~Q)^\kwtt= 
 (\vect{t}~\kwnot~\kwin~Q)^\kwff 
= \vect{t}~\kwin~Q^\kwtt }
\\
\multicolumn{3}{l}{
 (\vect{t}~\kwnot~\kwin~Q)^\kwtt = 
 (\vect{t}~\kwin~Q)^\kwff 
}
 \\
\multicolumn{3}{l}{
 \qquad =
   \kwnot~\kwex~(\kwsel~\ast~\kwfrom~[\mathit{table}~Q^\kwtt : \varphi(\vert \vect{t}\vert)]
}
\\
\multicolumn{3}{l}{
 \qquad \qquad  \kwwhere~\vect{(t^+_i~\kwis~\kwnull~\kwor~0.\varphi(\vert \vect{t} \vert)_i~\kwis~\kwnull~\kwor~t^+_1 = 0.\varphi(\vert \vect{t}\vert)_i)})
}
\\
\phantom{nothing}
\\
\multicolumn{3}{l}{
(\kwsel~[\kwdist]~\vect{t : x}~\kwfrom~G~\kwwhere~{c})^\kwtt 
=
\kwsel~[\kwdist]~\vect{t : x}~\kwfrom~G^\kwtt~\kwwhere~c
}
\\
\phantom{nothing}
\\
\multicolumn{3}{l}{
(\kwsel~[\kwdist]~\ast~\kwfrom~{\vect{T : \beta}}~\kwwhere~{c})^\kwtt 
=
\kwsel~[\kwdist]~\ast~\kwfrom~{\vect{T^\kwtt : \beta}}~\kwwhere~c
}
\\
\phantom{nothing}
\\
\multicolumn{3}{l}{
(Q_1~\{\kwunion \vert \kwinters \vert \kwexcept \}~Q_2)^\kwtt
= Q_1^\kwtt~\{\kwunion \vert \kwinters \vert \kwexcept \}~Q_2^\kwtt
}
\\
\phantom{nothing}
\\
\multicolumn{3}{l}{
G^\kwtt = ((\vect{T_1 : \beta_1})~\kwlat~\ldots~\kwlat~(\vect{T_k : \beta_k}))^\kwtt 
= (\vect{T_1^\kwtt : \beta_1})~\kwlat~\ldots~\kwlat~(\vect{T_k^\kwtt : \beta_k})
}
\end{array}
\]
}
\caption{Translation from 3VL-SQL to 2VL-SQL}\label{fig:translation}
\end{figure}

Figure~\ref{fig:translation} shows the definition of these transformations: these extend Guagliardo and Libkin's version by adding cases for |LATERAL| query inputs and for the |IS TRUE| test.
Most of the interesting things happen within conditions: while the definition of (|$\vect{t}$ IN Q|)$^\kwtt$ simply propagates the transformation to the nested query, the definition of (|$\vect{t}$ NOT IN Q|)$^\kwtt$ is more involved: it requires us to evaluate |Q|$^\kwtt$ as a nested query and then keep those tuples that are equal to $\vect{t}$ up to the presence of |NULL|s (either in $\vect{t}$ or in |Q|); if the resulting relation is not empty, the condition evaluates to true; in the formalization a "fold_right" operation is used to generate all the conditions on the elements of $\vect{t}$ and of the tuples from |Q|. The definition of this case is further complicated by the fact that the schema of $Q$ may not be well-formed, so we need to replace it with a new schema made of pairwise distinct names (generated on the fly by the $\varphi$ operation); furthermore, since in the translated query we use $\vect{t}$ inside a nested |SELECT *| query (thus, in an extended context), we use the "tm_lift" operation to increment the de~Bruijn indices it may contain (in the figure, we use the notation $t^+_i$ for this operation). Negations are translated as (|NOT $c$|)$^\kwtt = c^\kwff$; the transformation commutes in the other cases.

As for the negative translation $(\cdot)^\kwff$, it proceeds by
propagating the negation to the leaves of the conditional expression
(using de~Morgan's laws for |AND|s and |OR|s). The membership tests
(|$\vect{t}$ IN Q|)$^\kwff$ and (|$\vect{t}$ NOT IN Q|)$^\kwff$ are
defined as in the positive translation, but with their roles swapped. 
In the interesting case, we translate $P^n(\vect{t})^\kwff$ by checking that $P^n(\vect{t})$ is not true and that all elements of $\vect{t}$ are not null (here as well, the condition is computed by means of a "fold_right" on the elements of $\vect{t}$). The two translations are described by the following Coq code.

\begin{coq}
  Fixpoint ttcond (d: Db.D) (c : precond) : precond :=
    match c with
    | cndmemb true tl Q => cndmemb true tl (ttquery d Q)
    | cndmemb false tl Q => 
        let al := freshlist (length tl) in
          cndnot (cndex (selstar false 
            [(tbquery (ttquery d Q), al)]
            (List.fold_right (fun (ta : pretm * Name) acc =>
              let (t,a) := ta in
              cndand (cndor (cndnull true (tmvar (0,a))) 
                (cndor (cndnull true (tm_lift t 1)) 
                  (cndeq (tm_lift t 1) (tmvar (0,a))))) acc)
             cndtrue (List.combine tl al))))
    | cndex Q => cndex (ttquery d Q)
    | cndnot c1 => ffcond d c1
	(* ... *)
    end
  with ffcond (d: Db.D) (c : precond) : precond :=
    match c with
    | cndtrue => cndfalse
    | cndfalse => cndtrue
    | cndnull b t => cndnull (negb b) t
    | cndpred n p tml => 
        cndand (cndnot c) 
          (List.fold_right (fun t acc => 
            cndand (cndnull false t) acc) cndtrue tml)
    | cndmemb true tl Q => 
        let al := freshlist (length tl) in
          cndnot (cndex (selstar false 
            [(tbquery (ttquery d Q), al)]
            (List.fold_right (fun (ta : pretm * Name) acc =>
              let (t,a) := ta in
              cndand (cndor (cndnull true (tmvar (0,a))) 
                (cndor (cndnull true (tm_lift t 1)) 
                  (cndeq (tm_lift t 1) (tmvar (0,a))))) acc)
             cndtrue (List.combine tl al))))
    | cndmemb false tl Q => cndmemb true tl (ttquery d Q)
    | cndex Q => cndnot (cndex (ttquery d Q))
    | cndand c1 c2 => cndor (ffcond d c1) (ffcond d c2)
    | cndor c1 c2 => cndand (ffcond d c1) (ffcond d c2)
    | cndnot c1 => ttcond d c1
    end
  with ttquery (d: Db.D) (Q : prequery) : prequery :=
    match Q with
    | select b btm btb c => 
        select b btm (List.map (fun bt => 
          (tttable d (fst bt), snd bt)) btb) (ttcond d c)
    (* ... *)
    end
  with tttable (d: Db.D) (T : pretb) : pretb :=
    match T with
    | tbquery Q => tbquery (ttquery d Q)
    | _ => T
    end
  .
\end{coq}

We prove that the translation preserves the semantics of queries in the following theorem.
\begin{theorem} 
For all queries $Q$, if $\sem{\jq{\Gamma}{Q}{\tau}}^\threevl \eval S$, there exists $S'$ such that
\linebreak
$\sem{\jq{\Gamma}{Q^\kwtt}{\tau}}^\twovl \eval S'$ and for all $\eta : \mathtt{env}~\Gamma$, $S~\eta = S'~\eta$.
\end{theorem}
The proof of the theorem is by induction on the semantic judgments
yielding $S$: this is actually a mutual induction on the five mutually
defined evaluations. For the part of the proof that deals with
conditions, we need to prove a stronger statement that essentially says that $c^\kwtt$ evaluates to "true" only if  $c$ evaluates to "ttrue", and $c^\kwff$ evaluates to "true" only if $c$ evaluates to "tfalse": in other words, $c^\kwtt$ asserts the truth of $c$, while $c^\kwff$ asserts its falsehood.

An immediate question raised by this result asks whether a realistic
semantics for \SQLNull can be derived from a semantics that does not
have a special treatment of null values, just by translating input
queries under the the $(\cdot)^\kwtt$ transformation. The answer is
affirmative in principle: however, to prove the validity of rewrite
rules under that semantics, one would then need to reason not on the
original query $Q$, but on its translated version $Q^\kwtt$. This
would greatly complicate the proof since, recursively, one would
need to reason on conditions using two different induction
hypotheses for their positive and negative translation.

%% file: rc.tex
\section{Embedding the relational calculus}\label{sec:rc}
We now formalize a relational calculus to show that its normal forms
can be translated to SQL in a semantically preserving way. The
calculus we describe is a variant of the heterogeneous nested
relational calculus (\NRClsb~\cite{Ricciotti19dbpl,Ricciotti20fscd}),
which provides both set and bag semantics, enriched with a constant
$\kwnull$ to account for indeterminate values. All variants of \NRC
allow terms of nested collection type, which cannot be expressed in SQL directly; however, we 
will show that normal forms whose type is a flat relation can be translated to SQL. 

The terms of \NRClsb are defined by the following grammar:
\[
\begin{array}{lcl@{\hspace{1.2cm}}lcl@{\hspace{1.2cm}}lcl@{\hspace{1.2cm}}lcl}
M & ::=  & n \orelse \mathbf{k} \orelse \kwnull \orelse P^n(\vect{M_n}) \orelse \plempty_b(M) 
\\
  & \orelse & \kwtrue \orelse \kwfalse \orelse \kwisnull(M) \orelse \kwistrue(M) \orelse M_1 \land M_2 \orelse M_1 \lor M_2 \orelse \lnot M
\\
  & \orelse & \tuple{\vect{x = M}} \orelse M.x \orelse \kwtb~x
\\
  & \orelse & \emptyset_{b,\sigma} \orelse \setlit{M}_b \orelse \distinct M \orelse \promote M
\\
  & \orelse & M_1 \cup M_2 \orelse M_1 - M_2
              \orelse \comprehension{M_1 \mid M_2} \orelse M_1~\kwwhere~M_2

\end{array}
\]
Variables are represented as de~Bruijn indices $n$. The grammar
provides empty collections and singletons, along with the standard
operations of union, intersection, and difference; empty collections
$\emptyset$ and singletons $\setlit{M}$ are annotated with a subscript
$b$ representing their kind, which can be $\kwset$ or $\kwbag$; empty
collections are additionally annotated with their schema $\sigma$; the
other collection operations do not require annotations. There are also
operations $\distinct$ and $\promote$, which respectively convert a
bag into a set by duplicate elimination, and promote a set to a bag in
which each element has multiplicity equal to 1. A comprehension
$\comprehension{M_1 \mid M_2}$ binds a variable in $M_1$:
semantically, this corresponds to the union of the $M_1[V/0]$ for all
values $V$ in the collection $M_2$ (this is a set or bag union
depending on whether $M_1$ and $M_2$ are sets or bags); $M_1$ and
$M_2$ are called the head and the generator of a comprehension
respectively. The one-armed conditional $M_1~\kwwhere~ M_2$ is
equivalent to $M_1$ when $M_2$ is true, and to an empty collection
otherwise. The emptiness test $\plempty_b(M)$ is annotated with a
Boolean depending on whether its argument is a set or a bag.

Tuples with named fields $\tuple{\vect{x = M}}$, and tuple projections $M.x$ are standard; null values $\kwnull$, constants $\mathbf{k}$, standard Boolean operations and constants, the test for nullness $\kwisnull(M)$, the test for truth $\kwistrue(M)$, custom predicates $P^n(\vect{M_n})$, and table references $\kwtb~x$ are similar to the corresponding SQL concepts of Section~\ref{sec:syntax}. 

The abstract syntax above corresponds to the following Coq implementation.
\begin{coq}
  Inductive tm :=
  | cst     : BaseConst -> tm
  | null    : tm
  | pred    : forall n, (forall l : list BaseConst, length l = n -> bool) 
                -> list tm -> tm
  | rctrue  : tm
  | rcfalse : tm
  | isnull  : tm -> tm              (* $\kwisnull(M)$ *)
  | istrue  : tm -> tm              (* $\kwistrue(M)$ *)
  | rcand   : tm -> tm -> tm         (* $M \land N$ *)
  | rcor    : tm -> tm -> tm         (* $M \lor N$ *)
  | rcnot   : tm -> tm              (* $\lnot M$ *)
  | var     : nat -> tm
  | mktup   : list (Name * tm) -> tm
  | proj    : tm -> Name -> tm
  | tab     : Name -> tm
  | nil     : bool -> Scm -> tm       (* $\emptyset_{b,\sigma}$ *)
  | single  : bool -> tm -> tm        (* $\setlit{M}_b$ *)
  | union   : tm -> tm -> tm          (* $M \cup N$ *)
  | diff    : tm -> tm -> tm          (* $M - N$ *)
  | comprn  : tm -> tm -> tm          (* $\comprehension{M \mid N}$ *)
  | cwhere  : tm -> tm -> tm          (* $M~\kwwhere~N$ *)
  | dist    : tm -> tm               (* $\distinct M$ *)
  | prom    : tm -> tm               (* $\promote M$ *)
  | empty   : bool -> tm -> tm        (* $\plempty_b(M)$ *)
\end{coq}
The most important difference between this concrete syntax and the abstract one is that where the latter uses subscripts $\kwbag$, $\kwset$, the former employs a Boolean which is true for sets, and false for bags.

In this formalization, we are only interested in assigning meaning to RC \emph{normal forms}, corresponding to the terms in this grammar:
\[
	\begin{array}{rcll}
	M & ::= & \emptyset_{\kwbag,\sigma} \orelse \bigcup \vect{D} & \text{bag collections} \\
	D & ::= & \comprehension{\setlit{V}~\kwwhere~B \mid \vect{G}} & \text{bag comprehensions} \\
	G & ::= & \kwtb~t \orelse \promote L \orelse M - M' & \text{bag comprehension generators}\\
	L & ::= & \emptyset_{\kwset,\sigma} \orelse \bigcup \vect{D} & \text{set collections} \\
	C & ::= & \comprehension{\setlit{V}~\kwwhere~B \mid \vect{F}} & \text{set comprehensions} \\
	F & ::= & \distinct (\kwtb~t) \orelse \distinct(M - M)' & \text{set comprehension generators}\\
	V & ::= & n \orelse \tuple{\vect{x = X}} & \text{tuples} \\
	B & ::= & \kwtrue \orelse \kwfalse \orelse \kwisnull~X \orelse \kwistrue~B & \text{conditions} \\
	& \orelse & p^n(\vect{X}) \orelse \plempty_\kwbag(M) \orelse \plempty_\kwset(L) & \\
	& \orelse & B \land B' \orelse B \lor B' \orelse \lnot B & \\
	X & ::= & \mathbf{k} \orelse \kwnull \orelse n.x & \text{base expressions}
	\end{array}
\]

In Coq, we define normal forms by means of an inductive judgment described in Fig.~\ref{fig:RCjudg}. Similarly to the grammar, the judgment partitions normal forms in various categories depending on their type: base expressions, tuples with a certain schema $\sigma$ ($\kwtuple~\sigma$), conditional tests ($\kwcond$), and collections of tuples ($\kwcoll~b,\sigma$), where $b$ can be $\kwbag$ or $\kwset$. Collections in normal form are defined as unions of nested comprehensions, thanks to auxiliary categories $\kwdisj~ b,\sigma$ and $\kwgen~ b,\sigma$ representing respectively comprehensions and comprehension generators.

\begin{figure}
\scriptsize{
\begin{tabular}{c}
\multicolumn{1}{l}{
%
\fbox{Variables (\texttt{j\_var})}
\hspace{1cm}
\AxiomC{$x \notin \sigma$}
\UnaryInfC{$x \# \sigma \vdash x$}
\DisplayProof
\hspace{1cm}
\AxiomC{$x \neq y$}
\AxiomC{$\sigma \vdash x$}
\BinaryInfC{$y \# \sigma \vdash x$}
\DisplayProof
}
\\

\\
\multicolumn{1}{l}{\fbox{Base expressions (\texttt{j\_nbase})}}
\\
\rVS
\AxiomC{$\phantom{A}$}
\UnaryInfC{$\jrcb{\Gamma}{\mathbf{k}}$}
\DisplayProof
\hspace{.7cm}
\AxiomC{$\phantom{A}$}
\UnaryInfC{$\jrcb{\Gamma}{\kwnull}$}
\DisplayProof
\hspace{.7cm}
\AxiomC{$\Gamma(n) = \kwsome~\sigma$}
\AxiomC{$\sigma \vdash x$}
\BinaryInfC{$\jrcb{\Gamma}{M.x}$}
\DisplayProof
\\

\\

\multicolumn{1}{l}{
%
\fbox{Tuples (\texttt{j\_ntuple})}
\hspace{1cm}
\AxiomC{$\Gamma(n) = \kwsome~\sigma$}
\UnaryInfC{$\jrct{\Gamma}{n}{\sigma}$}
\DisplayProof
\hspace{1cm}
\AxiomC{(for all $i$: $\jrcb{\Gamma}{M_i}$)}
\AxiomC{$\kwnodup~\vect{x}$}
\BinaryInfC{$\jrct{\Gamma}{\tuple{\vect{x = M}}}{\vect{x}}$}
\DisplayProof
}
\\

\\
\multicolumn{1}{l}{\fbox{Collections (\texttt{j\_ncoll})}}
\\
\rVS
\AxiomC{$\phantom{A}$}
\UnaryInfC{$\jrcs{\Gamma}{\emptyset_{b,\sigma}}{b,\sigma}$}
\DisplayProof
\hspace{1cm}
\AxiomC{$\jrcd{\Gamma}{M}{b,\sigma}$}
\UnaryInfC{$\jrcs{\Gamma}{M}{b,\sigma}$}
\DisplayProof
\hspace{1cm}
\AxiomC{$\jrcs{\Gamma}{M}{b,\sigma}$}
\noLine\UnaryInfC{$\jrcs{\Gamma}{N}{b,\sigma}$}
\UnaryInfC{$\jrcs{\Gamma}{M \cup N}{b,\sigma}$}
\DisplayProof
\\

\\
\multicolumn{1}{l}{\fbox{Disjuncts (\texttt{j\_ndisj)}}}
\\
\rVS
\AxiomC{$\jrct{\Gamma}{M}{\sigma}$}
\noLine\UnaryInfC{$\jrcc{\Gamma}{N}$}
\UnaryInfC{$\jrcd{\Gamma}{\setlit{M}_b~\kwwhere~N}{b,\sigma}$}
\DisplayProof
\hspace{1cm}
\AxiomC{$\jrcg{\Gamma}{N}{b,\tau}$}
\noLine\UnaryInfC{$\jrcd{\tau\#\Gamma}{M}{b,\sigma}$}
\UnaryInfC{$\jrcd{\Gamma}{\comprehension{M \mid N}}{b,\sigma}$}
\DisplayProof
\\

\\
\multicolumn{1}{l}{\fbox{Generators (\texttt{j\_ngen})}}
\\
\rVS
\AxiomC{$D(t) = \kwsome~\sigma$}
\UnaryInfC{$\jrcg{\Gamma}{\kwtb~t}{\kwbag,\sigma}$}
\DisplayProof
\hspace{1cm}
\AxiomC{$\jrcs{\Gamma}{M}{\kwset,\sigma}$}
\UnaryInfC{$\jrcg{\Gamma}{\promote M}{\kwbag,\sigma}$}
\DisplayProof
\\
\rVS
\AxiomC{$\jrcs{\Gamma}{M}{\kwbag,\sigma}$}
\AxiomC{$\jrcs{\Gamma}{N}{\kwbag,\sigma}$}
\BinaryInfC{$\jrcg{\Gamma}{M-N}{\kwbag,\sigma}$}
\DisplayProof
\hspace{1cm}
\AxiomC{$D(t) = \kwsome~\sigma$}
\UnaryInfC{$\jrcg{\Gamma}{\distinct (\kwtb~t)}{\kwset,\sigma}$}
\DisplayProof
\\
\rVS
\hspace{1cm}
\AxiomC{$\jrcs{\Gamma}{M}{\kwbag,\sigma}$}
\AxiomC{$\jrcs{\Gamma}{N}{\kwbag,\sigma}$}
\BinaryInfC{$\jrcg{\Gamma}{\distinct (M-N)}{\kwset,\sigma}$}
\DisplayProof
\\

\\
\multicolumn{1}{l}{\fbox{Conditions (\texttt{j\_ncond})}}
\\
\rVS

\AxiomC{$\phantom{A}$}
\UnaryInfC{$\jrcc{\Gamma}{\{\kwtrue \vert \kwfalse\}}$}
\DisplayProof
\hspace{1cm}
\AxiomC{$\jrcb{\Gamma}{M}$}
\UnaryInfC{$\jrcc{\Gamma}{\kwisnull(M)}$}
\DisplayProof
\hspace{1cm}
\AxiomC{$\jrcc{\Gamma}{M}$}
\UnaryInfC{$\jrcc{\Gamma}{\kwistrue(M)}$}
\DisplayProof
\\
\rVS
\AxiomC{$\vert \vect{t} \vert = n$} 
\AxiomC{$\jtl{\Gamma}{\vect{t}}$}
\BinaryInfC{$\jrcc{\Gamma}{P^n(\vect{t})}$}
\DisplayProof
\hspace{1cm}
\AxiomC{$\Gamma \vdash M : \kwcoll~b$}
\UnaryInfC{$\jrcc{\Gamma}{\plempty_b(M)}$}
\DisplayProof
\\
\rVS
\AxiomC{$\jrcc{\Gamma}{M}$}
\AxiomC{$\jrcc{\Gamma}{N}$}
\BinaryInfC{$\jrcc{\Gamma}{M \{\land \vert \lor\} N}$}
\DisplayProof
\hspace{1cm}
\AxiomC{$\jrcc{\Gamma}{M}$}
\UnaryInfC{$\jrcc{\Gamma}{\lnot~M}$}
\DisplayProof
\end{tabular}
}
\caption{Relational Calculus normal forms.}\label{fig:RCjudg}
\end{figure}

\subsection{Semantics}
We provide semantic evaluation judgments for RC terms using the same
approach we presented in Section~\ref{sec:semantics} for SQL queries:
as shown in Figure~\ref{fig:rcsemtypes}, there is a separate judgment
for each of the syntactic categories of terms in normal form. All
terms are interpreted using 3VL rather than Boolean logic.

\begin{figure}[t]
\footnotesize{
\[
\begin{array}{llll}
\mbox{Base exp.} & \sem{\Gamma \vdash E} \eval S_E & \mbox{s.t. } S_E & : \mathtt{env}~\Gamma \to \mathtt{V}
\\
\mbox{Tuples} & \sem{\jrct{\Gamma}{L}{\sigma}} \eval S_{L} & \mbox{s.t. } S_{L} & : \mathtt{env}~\Gamma \to \mathtt{T}~\len{\sigma}
\\
\mbox{Collections} & \sem{\jrcs{\Gamma}{M}{b,\sigma}} \eval S_{M} & \mbox{s.t. } S_M & : \mathtt{env}~\Gamma \to \mathtt{R}~\len{\sigma}
\\
\mbox{Disjuncts} & \sem{\jrcd{\Gamma}{C}{b,\sigma}} \eval S_C & \mbox{s.t. } S_C & : \mathtt{env}~\Gamma \to \mathtt{R}~\len{\sigma}
\\
\mbox{Generators} & \sem{\jrcg{\Gamma}{G}{b,\sigma}} \eval S_G & \mbox{s.t. } S_G & : \mathtt{env}~\Gamma \to \mathtt{R}~\len{\sigma}
\\
\mbox{Conditions} & \sem{\jrcc{\Gamma}{c}} \eval S_c & \mbox{s.t. } S_c & : \mathtt{env}~\Gamma \to \mathtt{tribool}
\end{array}
\]
}
\caption{Formal semantics of the Relational Calculus (types).}\label{fig:rcsemtypes}
\end{figure}

The evaluation of a base expression maps an environment to a value; valuations of sequences of base expressions return tuples of values, with arity corresponding to the length of the sequence; similarly, the evaluation of an RC tuple returns a tuple of values, with arity corresponding to the length of the tuple schema. Collections (and the auxiliary categories of disjuncts and generators) are mapped to evaluations returning relations, whose arity matches the schema of the input expression. Finally, the evaluation of conditions returns a truth value from $\mathtt{tribool}$.

Simple attributes are defined in a schema rather than a context: their semantics $\sem{\tau \vdash x}$ maps an environment for the singleton context $[\tau]$ to a value. Similarly, the semantics of fully qualified attributes $\sem{\Gamma \vdash n.x}$ maps an environment for $\Gamma$ to a value. In both cases, the output value is obtained by lookup into the environment.

The evaluation of terms $\sem{\jt{\Gamma}{t}}$ returns a value for $t$ given a certain environment $\gamma$ for $\Gamma$. In our definition, terms can be either full attributes $n.x$, constants $\bfk$, or "NULL". We have just explained the semantics of full attributes; on the other hand, constants and "NULL"s are already values and can thus be returned as such. The evaluation of term sequences $\sem{\Gamma \vdash \vect{t}}$, given an environment, returns the tuple of values corresponding to each of the terms and is implemented in the obvious way.

\subsection{Conversion to SQL}
Finally, in Figure~\ref{fig:RCxlatetypes} and~\ref{fig:RCxlate}, we formalize type and definition of the translation of normal form RC terms to SQL expressions: just like the RC semantics, this definition comprises several mutually inductive judgments, following the structure of normal forms rather than that of general RC expressions: this allows us to translate base expressions to SQL terms, tuples to sequences of SQL terms, conditions to SQL conditions, and collections to SQL queries. Comprehension generators are translated to SQL tables (which can be database tables or inner queries to be used in the $\kwfrom$ clause of an external query). Finally, disjuncts must return the three clauses of a $\kwsel-\kwfrom-\kwwhere$ statement: these are returned separately as a triple (for technical reasons related to the fact that recursion is needed to collect all these items in the case of nested comprehensions), and it is up to the collection translation judgment to compose them into a single SQL statement.

\begin{figure}[t]
\footnotesize{
\[
\begin{array}{lllll}
\mbox{Base exp.} & \xlate{\Gamma \vdash E} = M & \mbox{s.t. } & M & : \mathtt{Sql.pretm}
\\
\mbox{Tuples} & \xlate{\jrct{\Gamma}{L}{\sigma}} = \vect{M} & \mbox{s.t. } & M & : \mathtt{list~Sql.pretm}
\\
\mbox{Collections} & \xlate{\jrcs{\Gamma}{M}{b,\sigma}} = M' & \mbox{s.t. } & M' & : \mathtt{Sql.prequery}
\\
\mbox{Disjuncts} & \xlate{\jrcd{\Gamma}{C}{b,\sigma}} = (\vect{N},c,\vect{T}) & \mbox{s.t. } & \vect{N} & : \mathtt{list~Sql.pretm} \\
  & & & c & : \mathtt{Sql.precond} \\
  & & & \vect{T} & : \mathtt{list~Sql.pretb}
\\
\mbox{Generators} & \xlate{\jrcg{\Gamma}{G}{b,\sigma}} = T & \mbox{s.t. } & T & : \mathtt{Sql.pretb}
\\
\mbox{Conditions} & \xlate{\jrcc{\Gamma}{c}} = c' & \mbox{s.t. } & c' & : \mathtt{Sql.precond}
\end{array}
\]
}
\caption{Relational Calculus translation to SQL (types).}\label{fig:RCxlatetypes}
\end{figure}

\begin{figure}
\scriptsize{
\begin{tabular}{c}

\multicolumn{1}{l}{\fbox{Base expressions (\texttt{j\_base\_x})}}
\\
\rVS
\AxiomC{$\phantom{A}$}
\UnaryInfC{$\xlate{\jrcb{\Gamma}{\mathbf{k}}} = \mathbf{k}$}
\DisplayProof
\hspace{.7cm}
\AxiomC{$\phantom{A}$}
\UnaryInfC{$\xlate{\jrcb{\Gamma}{\kwnull}} = \kwnull$}
\DisplayProof
\hspace{.7cm}
\AxiomC{$\phantom{A}$}
\UnaryInfC{$\xlate{\jrcb{\Gamma}{n.x}} = n.x$}
\DisplayProof
\\

\\

\multicolumn{1}{l}{
%
\fbox{Tuples (\texttt{j\_tuple\_x})}
%
\hspace{1cm}
\AxiomC{(for all $i$: $\xlate{\jrcb{\Gamma}{M_i}} = N_i$)}
\UnaryInfC{$\xlate{\jrct{\Gamma}{\tuple{\vect{x = M}}}{\vect{x}}} = \vect{N}$}
\DisplayProof
}
\\

\\
\multicolumn{1}{l}{\fbox{Collections (\texttt{j\_coll\_x})}}
\\
\rVS
\AxiomC{$\phantom{A}$}
\UnaryInfC{$\xlate{\jrcs{\Gamma}{\emptyset_{b,\sigma}}{b,\sigma}} = \mathtt{sql\_nil}~\sigma$}
\DisplayProof
\hspace{1cm}
\AxiomC{$\xlate{\jrcd{\Gamma}{M}{b,\sigma}} = (\vect{N},c,\vect{G})$}
\UnaryInfC{$\xlate{\jrcs{\Gamma}{M}{b,\sigma}} = \mathtt{sql\_select}~b~(\vect{N},\sigma)~\vect{G}~c$}
\DisplayProof
\\
\rVS
\AxiomC{$\xlate{\jrcd{\Gamma}{M_1}{b,\sigma}} = (\vect{N_1},c_1,\vect{G_1})$}
\noLine\UnaryInfC{$\xlate{\jrcs{\Gamma}{M_2}{b,\sigma}} = N_2$}
\UnaryInfC{$\xlate{\jrcs{\Gamma}{M_1 \cup M_2}{b,\sigma}} = (\mathtt{sql\_select}~b~(\vect{N_1},\sigma)~\vect{G}~c)~\kwunion_b~N_2$}
\DisplayProof
\\

\\
\multicolumn{1}{l}{\fbox{Disjuncts (\texttt{j\_disj\_x)}}}
\\
\rVS
\AxiomC{$\xlate{\jrct{\Gamma}{M}{\sigma}} = \vect{M'}$}
\noLine\UnaryInfC{$\xlate{\jrcc{\Gamma}{N}} = N'$}
\UnaryInfC{$\xlate{\jrcd{\Gamma}{\setlit{M}_b~\kwwhere~N}{b,\sigma}} = (\vect{M'},N',[])$}
\DisplayProof
\hspace{1cm}
\AxiomC{$\xlate{\jrcg{\Gamma}{N}{b,\tau}} = N'$}
\noLine\UnaryInfC{$\xlate{\jrcd{\tau\#\Gamma}{M}{b,\sigma}} = (\vect{M'}, P', \vect{R'})$}
\UnaryInfC{$\xlate{\jrcd{\Gamma}{\comprehension{M \mid N}}{b,\sigma}} = (\vect{M'}, P', (N'\#\vect{R'}))$}
\DisplayProof
\\

\\
\multicolumn{1}{l}{\fbox{Generators (\texttt{j\_gen\_x})}}
\\
\rVS
\AxiomC{$D(t) = \kwsome~\sigma$}
\UnaryInfC{$\xlate{\jrcg{\Gamma}{\kwtb~t}{\kwbag,\sigma}} = \kwtb~t$}
\DisplayProof
\hspace{1cm}
\AxiomC{$\xlate{\jrcs{\Gamma}{M}{\kwset,\sigma}} = M'$}
\UnaryInfC{$\xlate{\jrcg{\Gamma}{\promote M}{\kwbag,\sigma}} = \kwquery~M'$}
\DisplayProof
\\
\rVS
\AxiomC{$\xlate{\jrcs{\Gamma}{M}{\kwbag,\sigma}} = M'$}
\AxiomC{$\xlate{\jrcs{\Gamma}{N}{\kwbag,\sigma}} = N'$}
\BinaryInfC{$\xlate{\jrcg{\Gamma}{M-N}{\kwbag,\sigma}} = \kwquery~(M'~\kwexcept~\kwall~N')$}
\DisplayProof
\\
\rVS
\hspace{1cm}
\AxiomC{$D(t) = \kwsome~\sigma$}
\UnaryInfC{$\xlate{\jrcg{\Gamma}{\distinct (\kwtb~t)}{\kwset,\sigma}} = \kwquery~(\mathtt{sql\_distinct}~(\kwtb~t))$}
\DisplayProof
\\
\rVS
\hspace{1cm}
\AxiomC{$\xlate{\jrcs{\Gamma}{M}{\kwbag,\sigma}} = M'$}
\AxiomC{$\xlate{\jrcs{\Gamma}{N}{\kwbag,\sigma}} = N'$}
\BinaryInfC{$\xlate{\jrcg{\Gamma}{\distinct (M-N)}{\kwset,\sigma}} = \kwquery~(\mathtt{sql\_distinct}~(\kwquery~(M'~\kwexcept~\kwall~N')))$}
\DisplayProof
\\

\\
\multicolumn{1}{l}{\fbox{Conditions (\texttt{j\_cond\_x})}}
\\
\rVS
\AxiomC{$\phantom{A}$}
\UnaryInfC{$\xlate{\jrcc{\Gamma}{\kwtrue}} = \kwtrue$}
\DisplayProof
\hspace{1cm}
\AxiomC{$\phantom{A}$}
\UnaryInfC{$\xlate{\jrcc{\Gamma}{\kwfalse}} = \kwfalse$}
\DisplayProof
\\
\rVS
\AxiomC{$\xlate{\jrcc{\Gamma}{M}} = M'$}
\noLine\UnaryInfC{$\xlate{\jrcc{\Gamma}{N}} = N'$}
\UnaryInfC{$\xlate{\jrcc{\Gamma}{M \land N}} = M'~\kwand~N'$}
\DisplayProof
\hspace{1cm}
\AxiomC{$\xlate{\jrcc{\Gamma}{M}} = M'$}
\noLine\UnaryInfC{$\xlate{\jrcc{\Gamma}{N}} = N'$}
\UnaryInfC{$\xlate{\jrcc{\Gamma}{M \lor N}} = M'~\kwor~N'$}
\DisplayProof
\\
\rVS
\AxiomC{$\xlate{\jrcc{\Gamma}{M}} = M'$}
\UnaryInfC{$\xlate{\jrcc{\Gamma}{\lnot~M}} = \kwnot~ M'$}
\DisplayProof
\\
\rVS
\AxiomC{$\xlate{\jrcb{\Gamma}{M}} = M'$}
\UnaryInfC{$\xlate{\jrcc{\Gamma}{\kwisnull(M)}} = M'~\kwis~\kwnull$}
\DisplayProof
\hspace{1cm}
\AxiomC{$\xlate{\jrcc{\Gamma}{M}} = M'$}
\UnaryInfC{$\xlate{\jrcc{\Gamma}{\kwistrue(M)}} = M'~\kwis~\kwtrue$}
\DisplayProof
\\
\rVS
\AxiomC{$\vert \vect{t} \vert = n$} 
\AxiomC{$\xlate{\jtl{\Gamma}{\vect{t}}} = \vect{t'}$}
\BinaryInfC{$\xlate{\jrcc{\Gamma}{P^n(\vect{t})}} = P^n(\vect{t'})$}
\DisplayProof
\hspace{1cm}
\AxiomC{$\xlate{\Gamma \vdash M : \kwcoll~b,\sigma} = M'$}
\UnaryInfC{$\xlate{\jrcc{\Gamma}{\plempty_b(M)}} = \mathtt{sql\_empty}~M'~\sigma$}
\DisplayProof
\end{tabular}
}
\\
\begin{align*}
\mathtt{sql\_nil}~\vect{x} := &\ \kwsel~\vect{\kwnull : x}~\kwfrom~[]~\kwwhere~\kwfalse
\\
\mathtt{sql\_select}~\kwbag~(\vect{t,x})~\vect{G}~c := &\ \kwsel~\vect{t : x}~\kwfrom~G_1~\kwlat~\ldots~\kwlat~G_n~\kwwhere~c
\\
\mathtt{sql\_select}~\kwset~(\vect{t,x})~\vect{G}~c := &\ \kwsel~\kwdist~\vect{t : x}~\kwfrom~G_1~\kwlat~\ldots~\kwlat~G_n~\kwwhere~c
\\
\mathtt{sql\_distinct}~T~\sigma := &\ \kwsel~\kwdist~*~\kwfrom~T:\sigma~\kwwhere~\kwtrue
\\
\mathtt{sql\_empty}~M~\sigma := &\ \kwnot~\kwex~\kwsel~*~\kwfrom~(\kwquery~M) : \sigma~\kwwhere~\kwtrue
\end{align*}
\caption{Relational Calculus translation to SQL.}\label{fig:RCxlate}
\end{figure}

The translation rules use some additional definitions as useful shorthands: $\mathtt{sql\_nil}$ returns an SQL query returning an empty relation of a certain schema; $\mathtt{sql\_select}$ composes its input into a $\kwsel-\kwfrom-\kwwhere$ statement: an important point to note is that all the inputs to this query are declared as $\kwlat$ due to the fact that the in the relational calculus, in a nested comprehension of the form $\bigcup\setlit{\bigcup\setlit{L \mid M} \mid N}$, $M$ is allowed to reference the tuples in $N$: therefore, similar dependencies must be allowed in the output of the translation as well. Another auxiliary definition $\mathtt{sql\_distinct}$ uses $\kwsel~\kwdist~*$ to deduplicate an input table with a given schema; $\mathtt{sql\_empty}$ constructs an SQL condition which is true whenever a certain query evaluates to an empty relation.


We are able to prove that the translation above is correct, by showing
that the semantics of an RC collection expression is equal to that of
the corresponding SQL query:
\begin{theorem}
Suppose $\sem{\jrcs{\Gamma}{M}{b,\sigma}} \eval S_M$; then, for all $M'$ such that
\linebreak
$\xlate{\jrcs{\Gamma}{M}{b,\sigma}} = M'$, there exists $S_{M'}$ such that
$\sem{\jq{\Gamma}{M'}{\sigma}}^\threevl \eval S_{M'}$ and for all $\eta : \mathtt{env}~\Gamma$, we have $S_M~\eta = S_{M'}~\eta$.
\end{theorem}
The proof of the theorem is by induction on the semantic judgment yielding $S_M$, followed by inversion on the translation of $M$ to $M'$.
More precisely, the proof uses mutual induction on the four mutually defined judgments for the semantics of collections, disjuncts, generators, and conditions.

%% file: related.tex
\section{Related work}
\label{sec:related}

\subsubsection*{Semantics of query languages with incomplete information and nulls}
Nulls arise from the need for \emph{incomplete information} in
databases, which was appreciated from an early stage. Codd~\cite{Codd79}
made one of the first proposals based on null values and three-valued
logic, though it was criticized early on due to semantic
irregularities and remains a controversial
feature~\cite{Rubinson07,Grant08}.  A great deal of subsequent
research has gone into proposing semantically satisfying approaches
to incomplete information, in which a database with null values (or
other additional constructs) is viewed as representing a \emph{set of
  possible worlds}, and we wish to find \emph{certain} query answers
that are true in all possible worlds.  Many of these techniques are
surveyed by van der Meyden~\cite{vandermeyden98ldbis}, but most such techniques
either make query answering intractable (e.g. coNP-hard), have
semantic problems of their own, or both.  However, SQL's standard
behaviour remains largely as proposed by Codd, leading database
researchers such as Libkin~\cite{libkin14pods} to propose revisiting the
topic with an eye towards identifying \emph{principled} approaches to
incomplete information that are \emph{realistic} relative to the
standard capabilities of relational databases.  For example,
Libkin~\cite{Libkin16} compares certain answer semantics with SQL's actual
semantics, shows that SQL's treatment of nulls is neither sound nor
complete with respect to certain answers, and proposes modifications
to SQL's semantics that restore soundness or completeness while
remaining (like plain SQL) efficiently implementable.

Some
work has explored the semantics and logical properties of nulls in
set-valued relational queries, but did not grapple with SQL's
idiosyncrasies or multiset semantics~\cite{franconi12amw}.
Guagliardo and Libkin~\cite{guagliardo17} were the first to define a semantics that is a
realistic model of SQL's actual behaviour involving both multisets and
nulls.  They empirically validated a (Python) implementation of the
semantics against the behaviour of real database systems such as
PostgreSQL and MySQL, and confirmed some minor but nontrivial known
discrepancies between them in the process.  In addition they gave
(paper) proofs of the main results relating the SQL semantics,
three-valued and two-valued semantics.  Our work complements and
deepens this work by making all notions of their semantics precise and
formal, and formally proving their main result relating the
three-valued and two-valued semantics.

%

 Because our formalization follows Guagliardo and Libkin's on-paper
 presentation closely, it benefits indirectly from their extensive
 experimental validation.  Nevertheless, there remains a
small ``formalization gap'' between our work and theirs in the sense
that our (formally validated) Coq definitions might differ from their
(empirically validated) Python implementation.  So, in addition to
extending the coverage of SQL features as discussed below, it could be
worthwhile to derive an executable semantics from our definitions and
empirically validate it against the same examples they used.

\subsubsection*{Formalizations of query languages}
%

Malecha et al.~\cite{malecha10popl} formalized components of a relational database
engine (including a front-end providing a SQL-like relational core,
optimization laws including side-conditions, and an implementation of
B+-trees) in Coq using the YNot framework.  Their work (like most
prior formalizations) employs set semantics; while the data model
allows for fields to have optional types, the behaviour of missing
values in primitive operations is not discussed, and their semantics
is the standard two-valued, set-theoretic interpretation of relational
algebra.  The main technical challenge in this work was verifying the
correctness of imperative algorithms and pointer-based data structures
used in efficient database implementations.
Benzaken et al.~ \cite{benzaken14esop} formalized the relational data model, going
beyond the core relational operations in Malecha et al.'s
formalization to include integrity constraints (functional
dependencies).  They formalize a number of algorithms from database
theory whose standard presentations are imprecise, and showed that
careful attention to variable binding and freshness issues is
necessary to verify them.  Their formalization included proofs of
correctness of relational rewrite rules (with respect to the set-theoretic
semantics) but did not directly consider SQL queries, multiset
semantics, or features such as nulls.

Chu et al.~\cite{Chu17} presented a new approach to formalizing and reasoning
about SQL, called \HoTTSQL.  \HoTTSQL 
uses homotopy type theory to formalize SQL with multiset semantics,
correlated subqueries, and aggregation in Coq.  \HoTTSQL is based on
the intriguing insight (inspired by work on semiring-valued database
query semantics~\cite{Green07}) that we can define multisets as
\emph{functions} mapping tuples to cardinalities. They propose
representing cardinalities using certain (finite) \emph{types} thanks
to the univalence axiom; this means that Coq's strong support for
reasoning about types can be brought to bear, dramatically simplifying
many proofs of query equivalences.  However, since \HoTTSQL does not
consider nulls or three-valued logic, it validates query equivalences
that become unsound in the presence of nulls.  Unfortunately,
it does not
appear straightforward to extend the \HoTTSQL approach of conflating
types with semiring annotations to handle SQL-style three-valued logic
correctly.  In addition, the adequacy of \HoTTSQL's approach requires
proof. It should also be noted that the univalence axiom used by homotopy 
type theory and Streicher's K axiom required to work with John Major equality, 
which we used in our formalization, are incompatible: this would make it
challenging to merge the two efforts.


Most recently, Benzaken and Contejean~\cite{Benzaken19} proposed a formal semantics for a
subset of SQL (\SQLCoq) including all of the above-mentioned features: multiset
semantics, nulls, grouping and aggregation. 
SQL has well-known idiosyncrasies arising from interactions among these features: for example, the two queries
\begin{sql}
SELECT COUNT(field) FROM T
SELECT COUNT(*) FROM T
\end{sql}
are not equivalent.  The first one counts the number of
\emph{non-null} \verb|field| values in $T$, while the second counts the
number of rows, ignoring their (possibly null) values.  These two
queries \emph{are} provably equivalent in the \HoTTSQL semantics, but
are correctly handled by \SQLCoq.  

Moreover, Benzaken and Contejean
highlight the complexity of SQL's treatment of grouping and
aggregation for \emph{nested subqueries}, propose a semantics for such
queries, and prove correctness of translations from \SQLCoq to a
multiset-valued relational algebra \SQLAlg.  Their work focuses on bag
semantics and uses a Coq library for finite bags, and treats duplicate
elimination as a special case of grouping.
While grouping can be expressed, in principle,
by desugaring to correlated subqueries (an approach proposed by Buneman et al.~\cite{buneman94}
and adopted by \HoTTSQL, which we could also adapt to our setting) these features of \SQLCoq highlight
many intricacies of the semantics of grouping that make it difficult to get such a desugaring right.

We can highlight several aspects where our work
complements \SQLCoq: (1) superficially, their approach does not deal
with named aliases for table records, requiring additional renaming;
(2) their novel semantics is tested on example queries but not
evaluated as thoroughly as Guagliardo and Libkin's; (3) we present
well-formedness criteria for \SQLNull, which are more accurate than those considered
for \SQLCoq, ensuring that queries with unbound table references should not be accepted; 
(4) their work does not consider formal results such as
the equivalence of 2-valued and 3-valued semantics, which to the best
of our knowledge has not been investigated in the presence of grouping and
aggregation; (5) the fragment of SQL formalized in our work allows lateral joins, 
an SQL:1999 feature that is becoming increasingly popular thanks to the support by recent versions of major DBMSs;
(6) building on the support for lateral joins, we are able to formalize a verified translation from the 
nested relational calculus to SQL, which is of interest for the theory of programming languages supporting
language-integrated query. Finally, because of the complexity of
their semantics (required to handle SQL's idiosyncratic treatment of
grouping and aggregation), our formalization may be preferable for
proving properties of queries that lack these features; it would be
enlightening to formally relate our formalization with theirs, and
establish whether equivalences proved in \SQLNull are still valid in \SQLCoq.

Formalization has also been demonstrated to be useful for designing
and implementing new query languages and verified transformations, for
example in the QCert system~\cite{Auerbach17}.  This work considers a
nested version of relational calculus, and supports a subset of SQL as
a source language, but does not appear to implement
Guagliardo and Libkin's semantics for SQL nulls.  It could be
interesting to incorporate support for SQL-style nulls into such a
verified query compiler.

%% file: concl.tex
\section{Conclusion}
\label{sec:concl}


We have mechanically checked the recently proposed semantics of
\linebreak \SQLNull~\cite{guagliardo17} and proved the main results
about its metatheory.  Our work should be compared to two recent
formalizations, \HoTTSQL~\cite{Chu17}, and \SQLCoq~\cite{Benzaken19}.
Compared to \HoTTSQL, our representation of multisets is elementary
and it does not appear straightforward to adjust \HoTTSQL to handle
null values, since its treatment of predicates using homotopy type
theory assumes standard two-valued logic.  Compared to \SQLCoq, our
semantics is simpler and closely modeled on the on-paper semantics of
\cite{guagliardo17}, which was thoroughly tested against real database
implementations. Our work is also the first formalization of SQL to consider queries with lateral inputs.
On the negative side, compared to both \HoTTSQL and
\SQLCoq, our formalization does not attempt to handle grouping and
aggregation, but as a result it may be simpler and easier to use, when
these features are not needed.


%
In this paper we also presented the first ever mechanized proofs of
the expressive equivalence of two-valued and three-valued SQL queries,
the first ever verified translation of relational calculus queries to SQL queries,
and the correctness of rewrite rules that are valid for SQL's real semantics (including
multisets and nulls).  The diversity of recent approaches to
formalizing SQL also suggests that consolidation
and cross-fertilization of ideas among approaches may reap rewards,
to provide a strong foundation for exploring
verification of other key components of database systems.

%% file: appendix.tex
\section{Commented semantics of \SQLNull}
We give here a commented version of the formalized semantics of \SQLNull beyond what was possible to report on in the paper.
The semantics consists of four inductive judgments for simple attributes, full attributes, terms and terms sequences ("j_var_sem", "j_fvar_sem", "j_tm_sem", "j_tml_sem"), and five mutually defined judgment for the main SQL expressions, namely queries ("j_q_sem"), tables ("j_tb_sem"), conditions ("j_cond_sem"), table bindings ("j_btb_sem"), and existentially nested queries ("j_in_q_sem").

\subsection{Semantics of attributes}
An attribute is evaluated in a singleton context $[s]$ under an environment for that context.

\[
\sem{s \vdash a} \eval S_a
\qquad
S_a : \mathtt{env} \; [s] \to \mathtt{Value}
\]

\begin{coq}
  Inductive j_var_sem : 
    forall s, Name -> (env (s::List.nil) -> Value) 
	-> Prop :=
  | jvs_hd : forall a s, ~ List.In a s -> 
      j_var_sem (a::s) a (fun h => env_hd h)
  | jvs_tl : forall a s b, 
      forall Sb, a <> b -> j_var_sem s b Sb -> 
	  j_var_sem (a::s) b (fun h => Sb (env_tl h)).
\end{coq}
Under a context $a::s$, the semantics of $a$ is the head value in the environment; we also check that $a$ should not be in the remainder of the context for well-formedness. Under a context $a::s$ where $a \neq b$, we first evaluate $b$ under the remainder context $s$ and lift the resulting evaluation from context $s$ to context $a::s$.

The judgment for full variables (in the form $n.a$, where $n$ is a de Bruijn index) lifts the semantics of simple variables to contexts composed of multiple schemas.

\[
\sem{\Gamma \vdash n.a} \eval S_{n.a}
\qquad
S_{n.a} : \mathtt{env} \; \Gamma \to \mathtt{Value}
\]

\begin{coq}
  Inductive j_fvar_sem : 
    forall G, nat -> Name -> (env G -> Value) 
	-> Prop :=
  | jfs_hd : forall s G a, 
      forall Sa, j_var_sem s a Sa -> 
	  j_fvar_sem (s::G) O a 
	    (fun h => Sa (@subenv1 (s::List.nil) G h))
  | jfs_tl : forall s G i a,
      forall Sia, j_fvar_sem G i a Sia -> 
	  j_fvar_sem (s::G) (S i) a 
	    (fun h => Sia (@subenv2 (s::List.nil) G h)).
\end{coq}
To evaluate attributes in the form $0.a$ in a context $s::G$, we first evaluate the simple attribute $a$ in $s$ and then lift the resulting evaluation from $[s]$ to $s::G$. The evaluation of $(i+1).a$ is obtained recursively by evaluating $i.a$ in $G$ and lifting the valuation to $s::G$.

\subsection{Semantics of terms}
A term $t$ is evaluated in a context $\Gamma$ to a function from a suitable environment to values.

\[
\sem{\Gamma \vdash t} \eval S_{t}
\qquad
S_{t} : \mathtt{env} \; \Gamma \to \mathtt{Value}
\]

\begin{coq}
  Inductive j_tm_sem0 (G:Ctx) : 
	pretm -> (env G -> Value) 
	-> Prop :=
  | jts_const : forall c, 
      j_tm_sem0 G (tmconst c) (fun _ => Db.c_sem c)
  | jts_null  : 
      j_tm_sem0 G tmnull (fun _ => None)
  | jts_var   : forall i a, 
      forall Sia, j_fvar_sem G i a Sia -> 
	  j_tm_sem0 G (tmvar (i,a)) Sia.
\end{coq}
While the semantics of constants and nulls is trivial, full variables are evaluated in the judgment for full variables.

The evaluation of sequences of terms is similar, but it returns a tuple of values of corresponding size.
\[
\sem{\Gamma \vdash \vect{t}} \eval S_{\vect{t}}
\qquad
S_{\vect{t}} : \mathtt{env} \; \Gamma \to \mathtt{T} \; \vert \vect{t} \vert
\]

\begin{coq}
  Inductive j_tml_sem0 (G:Ctx) : 
    forall (tml : list pretm), 
	(env G -> Rel.T (List.length tml)) 
	-> Prop :=
  | jtmls_nil  : j_tml_sem0 G List.nil (fun _ => Vector.nil _)
  | jtmls_cons : forall t tml,
      forall St Stml,
      j_tm_sem0 G t St -> j_tml_sem0 G tml Stml ->
      j_tml_sem0 G (t::tml) (fun h => 
	    Vector.cons _ (St h) _ (Stml h)).
\end{coq}
This judgment is implemented in the obvious way, by mapping empty sequences of terms to an evaluation returning the empty tuple, and by recursion when the input list of terms is not empty.

\subsection{Semantics of queries}
If a query $Q$ with schema $\sigma$ is evaluated in a context $\Gamma$, we obtain a function returning a relation with arity corresponding to $\sigma$.

\[
\gsem{\jq{\Gamma}{Q}{\sigma}} \eval S_{q} \qquad S_{q} : \mathtt{env}~\Gamma \to \mathtt{R}~\vert \sigma \vert
\]
\begin{coq}
Inductive j_q_sem (d : Db.D) : 
  forall G (s : Scm), prequery -> 
  (env G -> Rel.R (List.length s))
  -> Prop := 
  | jqs_sel : forall G b tml btbl c,
      forall G0 Sbtbl Sc Stml s e,
      j_btbl_sem d G G0 btbl Sbtbl ->
      j_cond_sem d (G0++G) c Sc ->
      j_tml_sem (G0++G) (List.map fst tml) Stml -> 
      s = List.map snd tml ->
      j_q_sem d G s 
	    (select b tml btbl c) 
        (fun h => 
		   let S1 := Sbtbl h in
           let p  := fun Vl => Sem.is_btrue 
		     (Sc (Evl.env_app _ _ (Evl.env_of_tuple G0 Vl) h)) 
		   in
           let S2 := Rel.sel S1 p in
           let f  := fun Vl => Stml 
		     (Evl.env_app _ _ (Evl.env_of_tuple G0 Vl) h) in
           let S := cast _ _ e (Rel.sum S2 f)
           in if b then Rel.flat S else S)
\end{coq}
The evaluation of select queries was described in detail in the paper. Here, we just notice that the list "tml" contains pairs of terms and attribute names, where attribute names are used to produce the output schema. Since the Coq typechecker cannot automatically infer that the arity of the semantics for the list of terms "List.map fst tml" matches the arity of the schema "List.map snd tml", the rule takes evidence of this fact in the form of an equation "e". The output relation is flattened to a set relation if the |DISTINCT| clause (signaled by the boolean "b") was used.

\begin{coq}
  | jqs_selstar : forall G b btb c,
      forall G0 Sbtb Sc Stml e,
      j_btbl_sem d G G0 btb Sbtbl ->
      j_cond_sem d (G0++G) c Sc ->
      j_tml_sem (G0++G) (tmlist_of_ctx G0) Stml ->
      j_q_sem d G (List.concat G0) (selstar b btb c) 
        (fun h => let S1 := Sbtbl h in
           let p  := fun Vl => Sem.is_btrue 
            (Sc (Ev.env_app _ _ (Ev.env_of_tuple G0 Vl) h)) 
           in
           let S2 := Rel.sel S1 p in
           let f  := fun Vl => Stml 
            (Ev.env_app _ _ (Ev.env_of_tuple G0 Vl) h) in
           let S := cast _ _ e (Rel.sum S2 f)
           in if b then Rel.flat S else S)
\end{coq}
The evaluation of select star queries proceeds similarly, by desugaring the star to a list of terms ("tmlist_of_ctx G0").

\begin{coq}
  | jqs_union : forall G b q1 q2,
      forall s S1 S2,
      j_q_sem d G s q1 S1 -> j_q_sem d G s q2 S2 ->
      j_q_sem d G s (qunion b q1 q2) 
        (fun Vl => let S := Rel.plus (S1 Vl) (S2 Vl) 
          in if b then S else Rel.flat S)
  | jqs_inters : forall G b q1 q2,
      forall s S1 S2,
      j_q_sem d G s q1 S1 -> j_q_sem d G s q2 S2 ->
      j_q_sem d G s (qinters b q1 q2) 
        (fun Vl => let S := Rel.inter (S1 Vl) (S2 Vl) 
          in if b then S else Rel.flat S)
  | jqs_except : forall G b q1 q2,
      forall s S1 S2,
      j_q_sem d G s q1 S1 -> j_q_sem d G s q2 S2 ->
      j_q_sem d G s (qexcept b q1 q2) 
        (fun Vl => if b then Rel.minus (S1 Vl) (S2 Vl) 
           else Rel.minus (Rel.flat (S1 Vl)) (S2 Vl))
\end{coq}
|UNION|, |INTERSECT| and |EXCEPT| queries are implemented all in the same fashion, by evaluating their subqueries recursively and combining them with the relational operators $\oplus$, $\cap$ and $\setminus$ from the ADT.

When a query $Q$ is evaluated in a context $\Gamma$ as an existentially nested query, we obtain a function returning a Boolean denoting whether the resulting relation is non-empty.

\[
\gsem{\jiq{\Gamma}{Q}} \eval S_{Q} \qquad S_{Q} : \mathtt{env}~\Gamma \to \mathtt{bool}
\]
\begin{coq}
  with j_in_q_sem (d : Db.D) : 
    forall G, prequery -> (env G -> bool) 
	-> Prop :=
  | jiqs_sel : forall G b tml btb c,
      forall G0 Sbtb Sc Stml,
      j_btb_sem d G G0 btb Sbtb ->
      j_cond_sem d (G0++G) c Sc ->
      j_tml_sem (G0++G) (List.map fst tml) Stml ->
      j_in_q_sem d G (select b tml btb c) 
        (fun h => let S1 := Sbtb h in
                  let p  := fun Vl => Sem.is_btrue 
                    (Sc (Ev.env_app _ _ 
                      (Ev.env_of_tuple G0 Vl) h)) in
                  let S2 := Rel.sel S1 p in
                  let f  := fun Vl => Stml (Ev.env_app _ _ 
                    (Ev.env_of_tuple G0 Vl) h) in
                  let S := Rel.sum S2 f
                  in 0 <? Rel.card 
                   (if b then Rel.flat S else S))
  | jiqs_selstar : forall G b btb c,
      forall G0 Sbtb Sc,
      j_btb_sem d G G0 btb Sbtb ->
      j_cond_sem d (G0++G) c Sc ->
      j_in_q_sem d G (selstar b btb c) 
        (fun h => let S1 := Sbtb h in
                  let p  := fun Vl => Sem.is_btrue 
                    (Sc (Ev.env_app _ _ 
                      (Ev.env_of_tuple G0 Vl) h)) in
                  let S2 := Rel.sel S1 p in
                  0 <? Rel.card
                   (if b then Rel.flat S2 else S2))
  | jiqs_union : forall G b q1 q2,
      forall s S1 S2,
      j_q_sem d G s q1 S1 -> j_q_sem d G s q2 S2 ->
      j_in_q_sem d G (qunion b q1 q2) 
       (fun Vl => let S := Rel.plus (S1 Vl) (S2 Vl) 
        in 0 <? Rel.card (if b then S else Rel.flat S))
  | jiqs_inters : forall G b q1 q2,
      forall s S1 S2,
      j_q_sem d G s q1 S1 -> j_q_sem d G s q2 S2 ->
      j_in_q_sem d G (qinters b q1 q2) 
       (fun Vl => let S := Rel.inter (S1 Vl) (S2 Vl) 
        in 0 <? Rel.card (if b then S else Rel.flat S))
  | jiqs_except : forall G b q1 q2,
      forall s S1 S2,
      j_q_sem d G s q1 S1 -> j_q_sem d G s q2 S2 ->
      j_in_q_sem d G (qexcept b q1 q2) 
       (fun Vl => 0 <? Rel.card 
        (if b then Rel.minus (S1 Vl) (S2 Vl) 
         else Rel.minus (Rel.flat (S1 Vl)) (S2 Vl)))
\end{coq}
The implementation of the semantics of existentially nested queries mostly reflects, in a simplified way, the corresponding rules for general queries. At the end, the resulting relation is tested for non-emptiness by checking whether its cardinality is greater than zero or not.

\subsection{Semantics of tables}
The type of the semantics of tables is similar to that of the semantics of queries.
\[
\gsem{\jT{\Gamma}{T}{\sigma}} \eval S_{T} \qquad S_{T} : \mathtt{env}~\Gamma \to \mathtt{R}~\vert \sigma \vert
\]
\begin{coq}
  with j_tb_sem (d : Db.D) : 
  forall G (s : Scm), pretb -> 
  (env G -> Rel.R (List.length s)) 
  -> Prop :=
  | jtbs_base : forall G x,
      forall s (e : Db.db_schema d x = Some s), 
      j_tb_sem d G s (tbbase x) (fun _ => Db.db_rel e)
  | jtbs_query : forall G q,
      forall s h,
      j_q_sem d G s q h ->
      j_tb_sem d G s (tbquery q) h
\end{coq}
The definition is trivial: the data base provides semantics for stored named tables, whereas tables resulting from queries are evaluated by means of their judgment.

The type of the semantics of frames (sequences of tables) is as follows:
\[
\gsem{\jTl{\Gamma}{\vect{T:\sigma}}{\Gamma'}} \eval S_{\vect{T}} \qquad S_{\vect{T}} : \mathtt{env}~\Gamma \to \mathtt{R}~\vert \mathtt{concat} \; \Gamma' \vert
\]
\begin{coq}
  with j_btb_sem (d : Db.D) : 
    forall G G', list (pretb * Scm) -> 
    (env G -> Rel.R (list_sum 
	  (List.map (length (A:=Name)) G'))) 
	-> Prop :=
  | jbtbs_nil : forall G, 
      j_btb_sem d G List.nil List.nil (fun _ => Rel.Rone) 
  | jbtbs_cons : forall G T s' btb,
      forall s G0 ST Sbtb e,
      j_tb_sem d G s T ST ->
      j_btb_sem d G G0 btb Sbtb -> 
	  length s = length s' ->
      j_btb_sem d G (s'::G0) ((T,s')::btb) (fun Vl => 
	    cast _ _ e (Rel.times (ST Vl) (Sbtb Vl)))
\end{coq}
The sequence of tables is unfolded as in the case of terms. The base case for empty sequences returns the 0-ary relation "Rone", which is the neutral element for the cartesian product of relations; non null sequences are evaluated recursively, and the resulting semantics are combined by means of the relational operator $\times$. A cast is used to make the definition typecheck.

The type of the semantics of generators (lateral joins of frames) is as follows:
\[
\gsem{\jTl{\Gamma}{G}{\Gamma'}} \eval S_{\vect{T}} \qquad S_{\vect{T}} : \mathtt{env}~\Gamma \to \mathtt{R}~\vert \mathtt{concat} \; \Gamma' \vert
\]
\begin{coq}
  with j_btbl_sem (d : Db.D) : 
    forall G G', list (list (pretb * Scm)) -> 
	(env G -> Rel.R (list_sum 
	  (List.map (length (A:=Name)) G'))) 
	-> Prop :=
  | jbtbls_nil : forall G, 
      j_btbl_sem d G List.nil List.nil (fun _ => Rel.Rone) 
  | jbtbls_cons : forall G btb btbl,
      forall G0 G1 Sbtb Sbtbl e,
      j_btb_sem d G G0 btb Sbtb ->
      j_btbl_sem d (G0 ++ G) G1 btbl Sbtbl ->
      j_btbl_sem d G (G1 ++ G0) (btb::btbl)
        (fun h =>
         let Rbtb := Sbtb h in
         Rel.rsum Rbtb (fun Vl => 
		   cast _ _ e (Rel.times 
		     (Sbtbl (Evl.env_app _ _ (Evl.env_of_tuple _ Vl) h)) 
			 (Rel.Rsingle Vl))))
\end{coq}
The generator too is evaluated one frame at a time, ending with the 0-ary relation "Rone" in the case of the empty sequence of frames: however, unlike the evaluation of simple frames, in a generator we do not perform a simple cartesian product of the semantics of the components, because a certain frame may depend on the ones declared to its left. More details are provided in the paper.

\subsection{Semantics of conditions}
The evaluation of conditions returns a truth value in the abstract type $\bfB$.
\[
\gsem{\jc{\Gamma}{c}} \eval S_{c} \qquad S_{c} : \mathtt{env}~\Gamma \to \bfB
\]
\begin{coq}
  with j_cond_sem (d : Db.D) : 
    forall G, precond -> (env G -> B) 
	-> Prop :=
  | jcs_true : forall G, 
      j_cond_sem d G cndtrue (fun _ => btrue)
  | jcs_false : forall G, 
      j_cond_sem d G cndfalse (fun _ => bfalse)
  | jcs_null : forall G b t, 
      forall St,
      j_tm_sem G t St ->
      j_cond_sem d G (cndnull b t) (fun Vl => 
	    of_bool (match St Vl with 
		  None => b | _ => negb b end))
  | jcs_istrue : forall G c,
      forall Sc,
      j_cond_sem d G c Sc ->
      j_cond_sem d G (cndistrue c) (fun Vl => 
	    of_bool (Sem.is_btrue (Sc Vl)))
\end{coq}
The evaluation of |TRUE| and |FALSE| is trivial, returning the corresponding elements of type $\bfB$. To evaluate |t IS [NOT] NULL|, we first evaluate "t" and then check whether the evaluation returns null or not. Similarly, to evaluate |c IS TRUE|, we first evaluate "c" and then check whether the evaluation yields |btrue| or not.

\begin{coq}
  | jcs_pred : forall G n p tml,
      forall Stml e,
      j_tml_sem G tml Stml ->
      j_cond_sem d G (cndpred n p tml) (fun Vl => 
	    Sem.sem_bpred _ p (to_list (Stml Vl)) 
		  (eq_trans (length_to_list _ _ _) e))
\end{coq}
This is the evaluation of an "n"-ary basic predicate "p" applied to a list of terms "tml". We first obtain "Stml" as the evaluation function for "tml", then the evaluation for the basic predicate is a function that takes an environment "Vl" as input and returns the result of "p" applied to "(Stml Vl)". However, "p" expects to receive list of constants, while "(Stml Vl)" is a tuple that may contain nulls: so, we first convert the tuple to a list, and then we use the operation "sem_bpred" from the ADT for "B" to lift a predicate of type "list BaseConst -> bool" to one of type "list Value -> B".

\begin{coq}
  | jcs_memb : forall G b tml q, 
      forall s Stml Sq (e : length tml = length s),
      j_tml_sem G tml Stml ->
      j_q_sem d G s q Sq ->
      let e' := f_equal Rel.T e in 
      j_cond_sem d G (cndmemb b tml q) (fun Vl => 
        let Stt := Rel.sel (Sq Vl) (fun rl => 
          Vector.fold_right2 (fun r0 V0 acc => 
            acc && is_btrue (veq r0 V0))
           true _ rl (cast _ _ e' (Stml Vl))) in
        let Suu := Rel.sel (Sq Vl) (fun rl => 
          Vector.fold_right2 (fun r0 V0 acc => 
            acc && negb (is_bfalse (veq r0 V0)))
           true _ rl (cast _ _ e' (Stml Vl))) in
        let ntt := Rel.card Stt in
        let nuu := Rel.card Suu in
        if (0 <? ntt) then of_bool b
        else if (0 <? nuu) then bmaybe
        else of_bool (negb b))
\end{coq}
The evaluation of membership of a tuple within a nested query was discussed in the paper; in the concrete definition, the boolean "b" is used to differentiate between "IS IN Q" and "IS NOT IN Q". Casts are also added to make the definitions typecheck.

\begin{coq}
  | jcs_ex : forall G q,
      forall Sq,
      j_in_q_sem d G q Sq ->
      j_cond_sem d G (cndex q) (fun Vl => of_bool (Sq Vl))
  | jcs_and : forall G c1 c2,
      forall Sc1 Sc2,
      j_cond_sem d G c1 Sc1 -> j_cond_sem d G c2 Sc2 ->
      j_cond_sem d G (cndand c1 c2) 
	    (fun Vl => band (Sc1 Vl) (Sc2 Vl))
  | jcs_or : forall G c1 c2,
      forall Sc1 Sc2,
      j_cond_sem d G c1 Sc1 -> j_cond_sem d G c2 Sc2 ->
      j_cond_sem d G (cndor c1 c2) 
	    (fun Vl => bor (Sc1 Vl) (Sc2 Vl))
  | jcs_not : forall G c0,
      forall Sc0,
      j_cond_sem d G c0 Sc0 ->
      j_cond_sem d G (cndnot c0) (fun Vl => bneg (Sc0 Vl))
\end{coq}
|EXISTS Q| conditions are implemented by the existentially nested query judgment "j_in_q_sem". The remaining conditions implement logical connectives by means of the corresponding operations on the ADT of truth values.

\section{Commented semantics of the flat relational calculus}

\subsection{Semantics of base expressions}

The semantics of base expressions of flat relational terms in normal form has the following type:
\[
\sem{\Gamma \vdash E} \eval S_E \qquad  S_E : \mathtt{env}~\Gamma \to \mathtt{V}
\]
Sequences of base terms are also given a semantic evaluation judgment for convenience:
\[
\sem{\Gamma \vdash \vect{E}} \eval S_{\vect{E}} \qquad  S_{\vect{E}} : \mathtt{env}~\Gamma \to \mathtt{T}~\len{\vect{E}}
\]
\begin{coq}
  Inductive j_base_sem (d : Db.D) : 
    forall G (t :  tm), (env G -> Value) -> Prop :=
  | jbs_cst  : forall G c,
      j_base_sem d G (cst c) (fun _ => Db.c_sem c)
  | jbs_null : forall G, j_base_sem d G null (fun _ => None)
  | jbs_proj : forall G i a,
      forall Sia,
      j_fvar_sem G i a Sia -> j_base_sem d G (proj (var i) a) Sia.

  Inductive j_basel_sem (d : Db.D) : 
    forall G (tml : list tm), 
	(env G -> Rel.T (List.length tml)) -> Prop :=
  | jbls_nil  : forall G, j_basel_sem d G List.nil (fun _ => Vector.nil _)
  | jbls_cons : forall G t tml,
      forall St Stml,
      j_base_sem d G t St -> j_basel_sem d G tml Stml ->
      j_basel_sem d G (t::tml) (fun h => Vector.cons _ (St h) _ (Stml h)).
\end{coq}
In the relational calculus, base expressions serve the same purpose as SQL terms, and their semantics are analogous.

\subsection{Semantics of tuples}
The semantics of flat relational calculus tuples in normal form has the following type:
\[
\sem{\jrct{\Gamma}{L}{\sigma}} \eval S_{L} \qquad  S_{L} : \mathtt{env}~\Gamma \to \mathtt{T}~\len{\sigma}
\]
\begin{coq}
  Inductive j_tuple_sem (d : Db.D) : 
    forall G (t:tm) (s:Scm), 
	(env G -> Rel.T (List.length s)) -> Prop :=
  | jts_mktup : forall G al bl,
      forall e Sbl, List.NoDup al -> 
	  List.length al = List.length bl -> j_basel_sem d G bl Sbl ->
      j_tuple_sem d G (mktup (List.combine al bl)) al (cast _ _ e Sbl).
\end{coq}
Normal forms of tuples are sequences of pairs attribute name-base expression ("List.combine al bl"): their semantics is obtained by evaluating the sequence of base expressions ("bl"), with a cast to make the judgment typecheck.

\subsection{Semantics of conditions}
The type of the semantics of conditions is as follows:
\[
\sem{\jrcc{\Gamma}{c}} \eval S_c \qquad  S_c : \mathtt{env}~\Gamma \to \mathtt{tribool}
\]
\begin{coq}
  Inductive j_cond_sem (d : Db.D) : forall G (t:tm), (env G -> B) -> Prop :=
  | jws_empty : forall G q b n, 
      forall Sq, j_coll_sem d G q b n Sq ->
      j_cond_sem d G (empty b q) (fun h => sem_empty _ (Sq h))
\end{coq}
The semantics of the emptiness test on a collection "q" first evaluates "q" recursively, and then uses an auxiliary definition "sem_empty" that checks whether the resulting relation has cardinality equal to zero.
\begin{coq}
  | jws_pred : forall G n p tml,
     forall Stml e,
     j_basel_sem d G tml Stml ->
     j_cond_sem d G (pred n p tml) (fun Vl => 
	   Sem.sem_bpred _ p (to_list (Stml Vl)) 
	     (eq_trans (length_to_list _ _ _) e))
  | jws_true : forall G, j_cond_sem d G rctrue (fun _ => Sem.btrue)
  | jws_false : forall G, j_cond_sem d G rcfalse (fun _ => Sem.bfalse)
  | jws_isnull : forall G t,
      forall St, j_base_sem d G t St ->
      j_cond_sem d G (isnull t) (fun Vl => Sem.of_bool 
	    (match St Vl with None => true | _ => false end))
  | jws_istrue : forall G c,
      forall Sc, j_cond_sem d G c Sc ->
      j_cond_sem d G (istrue c) (fun Vl => Sem.of_bool (Sem.is_btrue (Sc Vl)))
  | jws_and : forall G c1 c2,
      forall Sc1 Sc2, j_cond_sem d G c1 Sc1 -> j_cond_sem d G c2 Sc2 ->
      j_cond_sem d G (rcand c1 c2) (fun Vl => Sem.band (Sc1 Vl) (Sc2 Vl))
  | jws_or : forall G c1 c2,
      forall Sc1 Sc2, j_cond_sem d G c1 Sc1 -> j_cond_sem d G c2 Sc2 ->
      j_cond_sem d G (rcor c1 c2) (fun Vl => Sem.bor (Sc1 Vl) (Sc2 Vl))
  | jws_not : forall G c,
      forall Sc, j_cond_sem d G c Sc ->
      j_cond_sem d G (rcnot c) (fun Vl => Sem.bneg (Sc Vl))
\end{coq}
The remaining cases of the semantics of relational calculus conditions closely correspond to \SQLNull conditions, and their semantics is similar.

\subsection{Semantics of collections}
The type of the semantics of collections is as follows:
\[
\sem{\jrcs{\Gamma}{M}{b,\sigma}} \eval S_{M} \qquad  S_M : \mathtt{env}~\Gamma \to \mathtt{R}~\len{\sigma}
\]
\begin{coq}
  with j_coll_sem (d : Db.D) : 
    forall G (t : tm) (b:bool) (s:Scm), 
	(env G -> Rel.R (List.length s)) -> Prop :=
  | jcs_nnil : forall G b s,
     List.NoDup s -> j_coll_sem d G (nil b s) b s (fun h => sem_nil _)
  | jcs_ndisj : forall G t b s,
     forall St,
     j_disjunct_sem d G t b s St ->
     j_coll_sem d G t b s St
  | jcs_nunion : forall G t1 t2 b s,
      forall St1 St2,
      j_disjunct_sem d G t1 b s St1 ->
      j_coll_sem d G t2 b s St2 ->
      let S := fun h => 
	    if b then Rel.flat (Rel.plus (St1 h) (St2 h))
		else Rel.plus (St1 h) (St2 h) in
      j_coll_sem d G (union t1 t2) b s S
\end{coq}
To evaluate "nil b s" (that is, an empty collection with schema "s", where the Boolean "b" specifies whether the collection is a set or a bag), we need to provide an empty relation with arity equal to the length of "s": this is returned by the function "sem_nil", which first builds a singleton containing a tuple of nulls of suitable length, and then filters it using the trivially false predicate.
If the collection is a disjunct, it is evaluated by a separate judgment; if it is a union "union t1 t2" (where "t1" is a disjunct and "t2" a collection), the two subterms are evaluated recursively and then their semantics are combined using "Rel.plus" (this is followed by a call to "Rel.flat" to perform deduplication if "b" is true, signalling the collection is a set).

\subsection{Semantics of disjuncts}
The type of the semantics of disjuncts is the following:
\[
\sem{\jrcd{\Gamma}{C}{b,\sigma}} \eval S_C \qquad  S_C : \mathtt{env}~\Gamma \to \mathtt{R}~\len{\sigma}
\]
\begin{coq}
  with j_disjunct_sem (d : Db.D) :
    forall G (t : tm) (b : bool) (s:Scm), 
	(env G -> Rel.R (List.length s)) -> Prop :=
  | jds_single : forall G b tup c,
      forall stup Stup Sc, 
	  j_tuple_sem d G tup stup Stup -> 
	  j_cond_sem d G c Sc ->
      j_disjunct_sem d G (cwhere (single b tup) c) b stup 
        (fun h => 
		  if Sem.is_btrue (Sc h) then Rel.Rsingle (Stup h) 
		  else sem_nil _)
  | jds_comprn : forall G q1 q2,
      forall b sq2 Sq2 sq1 Sq1 e,
      j_gen_sem d G q2 b sq2 Sq2 ->
      j_disjunct_sem d (sq2::G) q1 b sq1 Sq1 ->
      j_disjunct_sem d G (comprn q1 q2) b sq1 (fun h =>
        let f := fun (Vl : Rel.T (length sq2)) => 
		  Sq1 (env_app _ _ (Evl.env_of_tuple (sq2::List.nil) 
		    (cast _ _ e Vl)) h) in
        let S := Rel.rsum (Sq2 h) f in
        if b then Rel.flat S else S)
\end{coq}
A disjunct is either $\setlit{M}_b~\kwwhere~N$, where $M$ is a tuple and $N$ is a condition, or a comprehension whose head is a disjunct. In the first case, we evaluate $M$ and $N$ using the respective judgments: if $N$ evaluates to true, we use "Rel.Rsingle" to return a singleton relation containing a tuple corresponding to the semantics of $M$; otherwise we return an empty relation of appropriate arity using "sem_nil".

In the case of a comprehension $\bigcup\setlit{M \mid N}$, we first evaluate the generator $N$, then for each element $\vect{v}$ of the resulting relation we evaluate the semantics of $M$ in an environment extended with $\vect{v}$; finally, we take the take the disjoint union of all the resulting relations using "Rel.rsum" (this is followed by a deduplication step if we are evaluating a set rather than a bag.

\subsection{Semantics of generators}
The type of the semantics of generators is the following:
\[
\sem{\jrcg{\Gamma}{G}{b,\sigma}} \eval S_G \qquad  S_G : \mathtt{env}~\Gamma \to \mathtt{R}~\len{\sigma}
\]
\begin{coq}
  with j_gen_sem (d : Db.D) : 
    forall G (t : tm) (b : bool) (s : Scm), 
	(env G -> Rel.R (List.length s)) -> Prop :=
  | jgs_tab : forall G x,
      forall s (e : Db.db_schema d x = Some s),
      j_gen_sem d G (tab x) false _ (fun _ => Db.db_rel e)
\end{coq}
The evaluation of named tables is provided by the database through "Db.db_rel".
\begin{coq}
  | jgs_prom : forall G q,
      forall s Sq,
      j_coll_sem d G q true s Sq ->
      j_gen_sem d G (prom q) false s Sq
\end{coq}
A bag generator can be a set collection "q" promoted to a bag. Its semantics is trivial, resorting to a recursive evaluation of "q" as a set collection.
\begin{coq}
  | jgs_bagdiff : forall G q1 q2,
      forall s Sq1 Sq2,
      j_coll_sem d G q1 false s Sq1 -> j_coll_sem d G q2 false s Sq2 ->
      j_gen_sem d G (diff q1 q2) false s (fun h => Rel.minus (Sq1 h) (Sq2 h))
\end{coq}
A generator consisting of the bag difference between two collections "q1" and "q2" is evaluated by taking the semantics of "q1" and "q2" as collections, and using "Rel.minus" to obtain the relation corresponding to their difference.
\begin{coq}
  | jgs_dtab : forall G x,
      forall s (e : Db.db_schema d x = Some s),
      j_gen_sem d G (dist (tab x)) true s (fun _ => Rel.flat (Db.db_rel e))
  | jgs_ddiff : forall G q1 q2,
      forall s Sq1 Sq2,
      j_coll_sem d G q1 false s Sq1 -> j_coll_sem d G q2 false s Sq2 ->
      j_gen_sem d G (dist (diff q1 q2)) true s 
	    (fun h => Rel.flat (Rel.minus (Sq1 h) (Sq2 h))).
\end{coq}
The semantics of deduplicated tables and bag differences is similar to the non-deduplicated case, but uses "Rel.flat" to deduplicate the result.